\documentclass{article}[11pt]
 \usepackage{logic9}
 
 \usepackage[disable]{todonotes}
	\newcommand{\igw}[1]{\todo[color=green!50]{\small #1}}

\usepackage[ruled]{algorithm2e}
\SetKwInput{Cp}{Communication predicate}
\SetKwInput{Where}{where}
\SetKwProg{Send}{send}{}{}
\SetKwSty{texttt}
\SetAlCapSkip{1ex}

\newcommand{\fo}{\mathbf{ir}}
\newcommand{\fop}{{\mathbf{ir}+1}}
\newcommand{\gp}{\overline{\p}}

\newcommand{\mset}{\mathit{mset}}
\newcommand{\vho}{\vec \ho}

\newcommand{\dar}{\!\!\downharpoonright}

\newcommand{\every}{\mathtt{every}}
\newcommand{\ls}{\mathtt{ls}}
\newcommand{\lr}{\mathtt{lr}}

\newcommand{\ep}{\epsilon}

\newcommand{\nn}{\mathbb{N}}

\newcommand{\one}{\mathit{one}}
\newcommand{\onea}{\mathit{one}^{a}}
\newcommand{\oneb}{\mathit{one}^{b}}
\newcommand{\oneq}{\mathit{one}^{?}}
\newcommand{\oned}{\mathit{one}^{d}}
\newcommand{\bias}{\mathit{bias}}
\newcommand{\biasq}{\mathit{bias}^{?}}
\newcommand{\biasqa}{\mathit{bias}^{?}_a}
\newcommand{\spread}{\mathit{spread}}
\newcommand{\spreadq}{\mathit{spread}^{?}}

\newcommand{\solo}{\mathit{solo}}
\newcommand{\soloq}{\mathit{solo}^{?}}
\newcommand{\soloa}{\mathit{solo}^{a}}
\newcommand{\ho}{\mathsf{H}}
\newcommand{\lact}[1]{\stackrel{#1}{\Longrightarrow}}

\newcommand{\bi}{\begin{itemize}}
\newcommand{\ei}{\end{itemize}}
\newcommand{\iif}{\mathtt{if}}

\newcommand{\tthen}{\mathtt{then}}
\newcommand{\cond}{\mathit{cond}}

\newcommand{\Dq}{D_{?}}

\newcommand{\mult}{\mathtt{mult}}

\newcommand{\uni}{\mathtt{uni}}
\newcommand{\thr}{\mathit{thr}}
\newcommand{\bthr}{\bar{\thr}}
\newcommand{\thrm}{\thr_m}
\newcommand{\thru}{\thr_u}
\newcommand{\thruo}{\thr_u^1}

\newcommand{\thrmok}{\thr_m^{1,k}}

\newcommand{\dom}{\mathit{set}}
\newcommand{\op}{\mathtt{op}}

\newcommand{\fire}{\operatorname{fire}}

\newcommand{\send}{\mathtt{send}}
\newcommand{\update}{\mathtt{update}}

\newcommand{\inp}{\mathit{inp}}
\newcommand{\dec}{\mathit{dec}}

\newcommand{\maxts}{\mathrm{maxts}}
\newcommand{\smor}{\mathrm{smor}}

 \theorembodyfont{\upshape}
  \theoremstyle{plain}
  \newtheorem{assumption}{Assumption}
  \newtheorem{proviso}{Proviso}

\title{Characterizing consensus in the Heard-Of model} 

 \author{A.R. Balasubramanian\footnote{Technical University of Munich, Germany} \and Igor Walukiewicz\footnote{CNRS, LaBRI, University of Bordeaux, France}}

\begin{document}
\setcounter{footnote}{2}
\maketitle

\begin{abstract}
	The Heard-Of model is a simple and relatively expressive model of distributed
	computation. 
	Because of this, it has gained a considerable attention of the verification community.
	We give a characterization of all algorithms solving consensus in a fragment
	of this model. 
	The fragment is big enough to cover many prominent consensus algorithms.
	The characterization is purely syntactic: it is expressed in terms of some
	conditions on the text of the algorithm. 
	One of the recent methods of verification of distributed algorithms is to
	abstract an algorithm to the Heard-Of model and then to verify the abstract algorithm
	using semi-automatic procedures. 
	Our results allow, in some cases, to avoid the  second step in this methodology.
\end{abstract}

\section{Introduction}

Most distributed algorithms solving problems like  consensus, leader election,
set agreement, or renaming are essentially one iterated loop. 
Yet, their behavior is difficult to understand due to unbounded number of processes,
 asynchrony, failures,  and other aspects of the execution model. 
The general context of this work is to be able to say what happens when we
change some of the parameters: modify an algorithm or the execution model. 
Ideally we would like to characterize the space of all algorithms solving a
particular problem.

To approach this kind of questions, one needs to restrict to a well defined space of all
distributed algorithms and execution contexts. 
In general this is an impossible requirement. 
Yet the distributed algorithms community has come up with some settings that are
expressive enough to represent interesting cases and limited enough to start
quantifying over ``all possible'' distributed algorithms~\cite{charron-heard-distributed09,WidSch:09,AguDelFau:12}. 

In this work we consider the consensus problem in the Heard-Of model~\cite{charron-heard-distributed09}. 
\emph{Consensus problem} is a central problem in the field of distributed algorithms;
it requires that all correct processes eventually decide on one of the initial
values. 
\emph{The Heard-Of model} is a round- and message-passing-based model.
It can represent many intricacies of various execution models  and yet is simple
enough to attempt to analyze it
algorithmically~\cite{ChaDebMer:11,DebMer:12,DraHenVei:14,MarSprBas:17,Mar:17}.   
Initially, our goal was to continue the quest from~\cite{MarSprBas:17} of examining what is
algorithmically possible to verify in the Heard-Of model. 
While working on this problem we have realized that a much more ambitious goal can
be achieved: to give a simple, and in particular decidable, characterization of all consensus
algorithms in well-defined fragments of the Heard-Of model.

The Heard-Of model is an open ended model: it does not specify what operations
processes can perform and what kinds of communication predicates are allowed. 
Communication predicates in the Heard-Of model capture in an elegant way both
synchrony degree and failure model.  
In this work we fix the set of atomic communication predicates and atomic
operations. 
We opted for a set sufficient to express most prominent consensus
algorithms (cf.\ Section~\ref{sec:examples}), but we do not cover all
operations found in the literature on the Heard-Of model.

Our characterization of algorithms that solve consensus is expressed in terms of
syntactic  conditions both on the text of the algorithm, and in the constraints
given by the communication predicate.
It exhibits an interesting way all  consensus algorithms should behave.
One could imagine that there can be a consensus algorithm that makes processes gradually
converge to a consensus: more and more processes adopting the same value. 
This is not the case.
A consensus algorithm, in models we study here, should
have a fixed number of crucial rounds where precise things are guaranteed to
happen. 
Special rounds have been identified for existing algorithms~\cite{RutMilSch:10},
but not their distribution over different phases.
Additionally, here we show that all algorithms should have this structure.


As an application of our characterization we can think of using it as an
intermediate step in analysis of more complicated settings than the Heard-Of model. 
An algorithm in a given setting can be abstracted to an algorithm in 
the Heard-Of model, and then our characterization can be applied. 
Instead of proving the original algorithm correct it is enough to show that the
abstraction is sound. 
For example, an approach reducing asynchronous semantics to round based semantics under some
conditions is developed in~\cite{ChaChaMer:09}. 
A recent paper~\cite{DamDraMil:19} gives a reduction methodology in a much
larger context, and shows its applicability. 
The goal language of the reduction is  an extension of the Heard-Of model that
is not covered by our characterization. 
As another application,  our characterization can be used to quickly see if an algorithm
can be improved by taking a less constrained communication predicate, by adapting
threshold constants, or by removing parts of code (c.f.\ Section~\ref{sec:examples}).

\subsubsection*{Related work}
The celebrated FLP result~\cite{FisLynPat:85} states
that consensus is impossible to achieve in an asynchronous system, even in the
presence of a single failure. 
There is a considerable literature investigating the models in which the consensus
problem is solvable.
Even closer in spirit to the present paper are results on weakest failure
detectors required to solve
the problem~\cite{ChaHadTou:96,FreGueKuz:11}.
Another step closer are works providing generic consensus algorithms that can be
instantiated to give several known concrete algorithms~\cite{MosRay:99,HurMosRay:02,GueRay:07,BieWidCha:07,SonRenSch:08,RutMilSch:10}.
The present paper considers a relatively simple model, but gives a characterization
result of all possible consensus algorithms.

The cornerstone idea of the  Heard-Of model is that both asynchrony and failures can
be modeled by the constraints on the message loss captured  by a notion of
communication predicates. 
This greatly simplifies the model that is essential for a kind of
characterization we present here. 
Unavoidably, not all aspects of partial synchrony~\cite{DwoLynSto:88,CriFet:99}
or failures~\cite{ChaTou:96} are covered by the model. 
For example, after a crash it may be difficult for a process to get into initial
state, or in terms of the Heard-of model, do the same round as other
processes~\cite{RenSchSch:15,ChaChaMer:09}.
These observations just underline that there is no universal model for
distributed algorithms. 
There exists several other proposals of relatively simple and expressible models~\cite{Gaf:98,WidSch:09,AguDelFau:12,RaySta:13}. 
The Heard-Of model, while not perfect, is in our opinion representative enough to
study in more detail.

On the verification side there are at least three approaches to analysis of
the Heard-Of or similar models. 
One is to use automatic theorem provers, like Isabelle~\cite{ChaMer:09,ChaDebMer:11,DebMer:12}.
Another is deductive verification methods applied to annotated programs~\cite{DraHenZuf:16,DraHenVei:14}.
The closest to this work is a model-checking approach~\cite{TsuSch:11,MarSprBas:17,Mar:17,AmiRubSto:18}.
Particularly relevant here is the work of Maric et al.~\cite{MarSprBas:17}. who
show cut-off results
for a fragment of the Heard-Of model and then perform verification on a resulting
finite state system. 
Our fragment of the Heard-Of model is incomparable with the one
from that work, and arguably it has less restrictions coming from
purely technical issues in proofs. 
While  trying to extend the scope of automatic methods we have realized that
we could actually bypass them completely and get a stronger characterization result.

Of course there are also other models of distributed systems that are considered
in the context of verification. 
For example there has been big progress on verification of threshold
automata~\cite{KukKonWid:18,KonLazVei:17,KonVeiWid:17,StoKonWid:19,BerKonLaz:19}.
There are also other methods, as  automatically generating invariants for
distributed algorithms~\cite{KonVeiWid:15,GleBjoRyb:16,TauLosMcM:18}, or
verification in Coq proof assistant~\cite{WilWooPan:15,WooWilAnt:16}.

\subsubsection*{Organization of the paper}
In the next section we introduce the Heard-Of model and formulate the consensus
problem.
In the four consecutive sections we present the characterizations for the
core model as well as for the extensions with timestamps, coordinators, as both
timestamps and coordinators at the same time.
We then give examples of algorithms that are covered by this model,
and discuss their optimality given our characterization.
The next for sections contain the proofs for the four characterizations.

\section{Heard-Of model and the consensus problem}

In a Heard-Of model a certain number of processes execute the same
code synchronously.
An algorithm consists of a sequence of \emph{rounds}, every process executes the same
round at the same time. 
The sequence of rounds, called \emph{phase}, is repeated forever.
In a round every process sends the value of one of its variables to a
communication medium, receives a multiset of values, and uses it to adopt a
new value (cf.\ Figure~\ref{fig:schema}).
\begin{figure}[htb]
	\label{fig:schema}
	\centering
	\includegraphics[scale=.5]{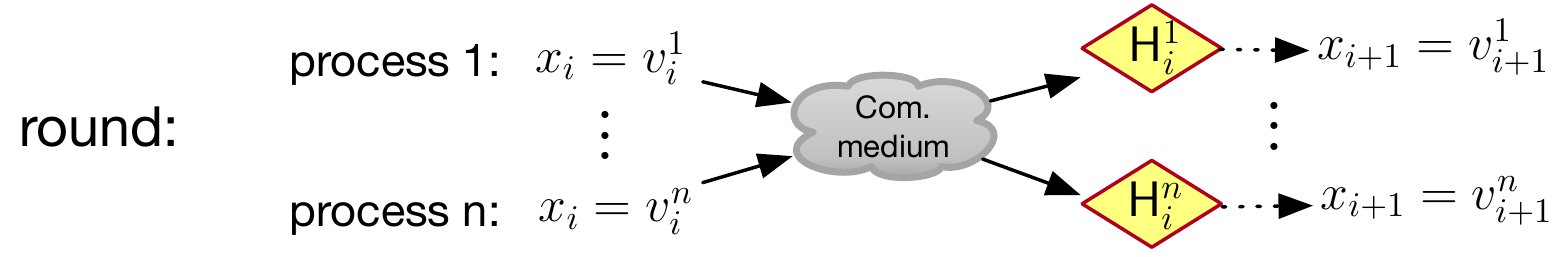}
	\bigskip

	\includegraphics[scale=.34]{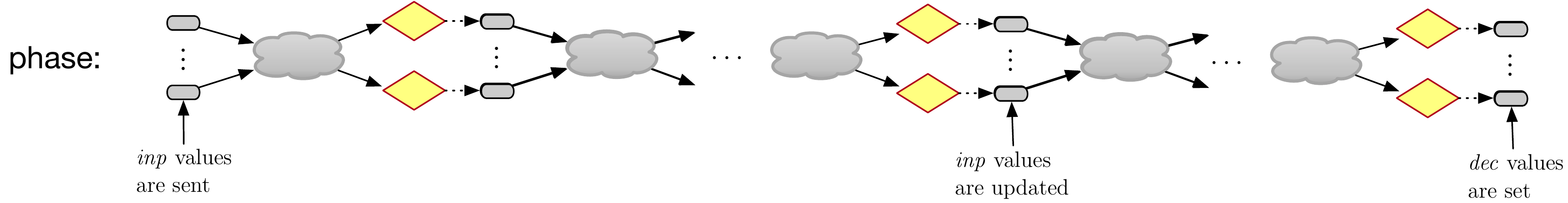}
	\caption{A schema of an execution of a round and of a phase. In every round $i$
	every process sends a value of its variable $x_i$, and sets its variable $x_{i+1}$
	depending on the received multiset of values: $\ho^j_i$.
	At the beginning of the phase the value of $\inp$
	is sent, at some round $\inp$ may be updated; we use $\fo$ for the index of
	this round.
	In the last round $\dec$ may
	be set. Both $\inp$ and $\dec$ are not updated if the value is $?$, standing
	for undefined.}
\end{figure}

At the beginning every process has its initial value in variable $\inp$.
Every process is expected to eventually set its decision variable $\dec$.
Every round is communication closed meaning that a value sent in a round can only be
received in the same round; if it is not received it is lost.
A \emph{communication predicate} is used to express a constraint on acceptable
message losses. Algorithm~\ref{alg:one-third} is a concrete simple example of
a $2$-round algorithm.

We proceed with a description of the syntax and semantics of  Heard-Of
algorithms. 
Next we define the consensus problem.
In later sections we will extend the core language with timestamps and
coordinators.

\begin{algorithm}[H]\label{alg:one-third}
		\Send{$(\inp)$}{
		\lIf{$\uni(\ho) \land |\ho| > \thr_1 \cdot |\Pi|$}{$x_1:=\inp:=\smor(\ho)$}
		\lIf{$\mult(\ho) \land |\ho| > \thr_1 \cdot |\Pi|$}{$x_1:=\inp:=\smor(\ho)$}
	}
	\Send{$x_1$}{
		\lIf{$\uni(\ho) \land |\ho| > \thr_2 \cdot |\Pi|$}{$\dec:=\smor(\ho)$}
	}
	\BlankLine
	\Cp{$\lF(\p^1 \land \lF\p^2)$}
	\Where{$\p^1 := (\f_{=}\land\f_{\thr_1},\true)$\quad and\quad $\p^2 := (\f_{\thr_1},\f_{\thr_2})$}
	\caption{Parametrized OneThird algorithm~\cite{charron-heard-distributed09}, $\thr_1, \thr_2$ are constants from $(0,1)$}
\end{algorithm}

\subsubsection*{Syntax}
An algorithm has one \emph{phase} that consists of two or more rounds. 
In the first round each process  sends the value of $\inp$ variable, in the last round it
can set the value of $\dec$ variable. 
A phase is repeated forever, all processes execute the same round at the same time.
A round $i$ is a send statement followed by a sequence of conditionals:

\RestyleAlgo{plain}
\begin{algorithm}[H]
		\Send{$x_{i-1}$}{
		\lIf{$\cond_i^1(\ho)$}{$x_i:=\op_i^1(\ho)$}
		\vdots
		\lIf{$\cond_i^l(\ho)$}{$x_i:=\op_i^l(\ho)$}
		}
\end{algorithm}
\noindent The variables are used in a sequence: first $x_0$, which is $\inp$, is sent and
$x_1$ is set, then $x_1$ is sent and $x_2$ is set, etc.\ (cf. Figure~\ref{fig:schema}).
There should be exactly one round (before the last round) where $\inp$ is
updated;  the conditional lines in this round are:
\begin{align*}
		\iif\ \cond^j_{\fo}(\ho)\ \tthen\ x_{\fo}:=\inp:=\op^j_{\fo}(\ho)
\end{align*}
Since this is a special round, we use the index $\fo$ to designate this round number. 
In the last round, only instructions setting variable $\dec$ can be present:
\begin{align*}
	   &\iif\ \cond^j_r(\ho)\ \tthen\ \dec:=\op^j_r(\ho)
\end{align*}
This is why a phase needs to have at least two rounds. 
Of course one can also have a syntax and a characterization for one round
algorithms, but unifying the two hinders readability.
Our fragment roughly corresponds to the fragment from~\cite{MarSprBas:17},
without extra restrictions but with a less liberty at the fork point. 

	%
As an example, consider Algorithm~\ref{alg:one-third}. 
It has two rounds, each begins with a $\send$ statement. 
In the first round both $x_1$ and $\inp$ are set, in the second round $\dec$ is set. 
The conditions talk about properties of the received $\ho$ multiset, that we
describe below. 

As the above syntax suggests, in round $i$ every process first sends the value of 
variable $x_{i-1}$, and then receives a multiset of values $\ho$ that it uses to
set the value of the variable $x_i$.
The possible tests on the received set $\ho$ are $\uni$, $\mult$, and $|\ho|>\thr\cdot |\Pi|$
saying respectively that: the multiset has only one value; has more than one
value; and that is of size $>\thr\cdot n$ where $n$ is the number of processes
and  $0\le\thr<1$. 
The possible operations are $\min(\ho)$ resulting in the minimal value in $\ho$,
and $\smor(\ho)$ resulting in the minimal most frequent value in $\ho$. 
For example, the first conditional line in Algorithm~\ref{alg:one-third} tests if
there is only one value in $\ho$, and if this value has multiplicity at least
$\thr_1\cdot n$ in $\ho$; if so $\inp$ and $x_1$ are set to this value, it does not
matter if $\min$ or $\smor$ operation is used in this case. 

In addition to description of rounds, an algorithm has also   a \emph{communication
predicate} putting constraints on the behavior of the communication
medium. 
A \emph{communication predicate for a phase} with $r$ rounds is a tuple
$\p=(\p_1,\dots,\p_r)$, where each $\p_l$ is a conjunction of atomic
communication predicates that we specify later. 
A \emph{communication predicate for an algorithm} is
\begin{equation*}
  (\lG\gp)\land(\lF(\p^1\land\lF(\p^2\land\dots (\lF\p^k)\dots)))
\end{equation*}
where $\gp$ and $\p^i$ are communication predicates for a phase.
Predicate $\gp$ is \emph{global predicate}, and $\p^1\dots,\p^k$ are \emph{sporadic
predicates}.
So the global predicate specifies constraints on every
phase of execution, while sporadic predicates specify a \emph{sequence} of special
phases that should  happen: first $\p_1$, followed later by $\p_2$, etc.
We have two types of atomic communication predicates: $\f_=$ says that every
process receives the same multiset; $\f_\thr$ says that every process receives a
multiset of size at least $\thr\cdot n$ where $n$ is the number of processes. 
In Algorithm~\ref{alg:one-third} the global predicate is trivial, and we require
two special phases.
In the first of them, in its first round every process should
receive exactly the same $\ho$ multiset, and the multiset should contain values from at least $\thr_1$
fraction of all processes.

\subsubsection*{Semantics}
The values of variables come from a fixed linearly ordered set $D$. 
Additionally, we take a special value $? \notin D$ standing
for undefined.
We write $\Dq$ for $D\cup\set{?}$.

We describe the semantics of an algorithm for $n$ processes.
A \emph{state of an algorithm} is a pair of $n$-tuples of values; denoted $(f,d)$.
Intuitively, $f$ specifies the value of the $\inp$ variable for each process, and
$d$ specifies the value of the $\dec$ variable. 
The value of $\inp$ can never be $?$, while initially the value of $\dec$ is $?$
for every process. 
We denote by $\mset(f)$ the multiset of values appearing in the tuple $f$, and
by $\dom(f)$ the set of values in $f$.
Only values of $\inp$ and $\dec$ survive between phases. 
All the other variables are reset to $?$ at the beginning of each phase.

There are two kinds of transitions:
\begin{equation*}
	(f,d)\act{\p} (f',d') \quad\text{and}\quad  f\lact{\f}_i f' \ .
\end{equation*}
The first is a phase transition, while the second is a transition for round $i$.
So in a transition of the second type $f$ describes the values of $x_i$, and
$f'$ the values of $x_{i+1}$.
Phase transition is labeled with a phase communication predicate, while a round
transition has a round number and a conjunction of atomic predicates as
labels.

Before defining these transitions we need to describe the semantics of
communication predicates. 
At every round processes send values of their variable to a communication medium, and then
receive a multiset of values from the medium (cf.\ Figure~\ref{fig:schema}).
Communication medium is not assumed to be perfect, it can send a different
multiset of values to every process, provided it is a sub-multiset of received
values.
An atomic communication predicate puts constraints on multisets that every
process receives. 
So a predicate specifies constraints on a tuple of multisets $\vho=(\ho_1,\dots,\ho_n)$.
Predicate $\f_=$ is satisfied if all the multisets are the same.
Predicate $\f_\thr$ requires that every multiset is bigger than $\thr\cdot n$
for some number $0\le\thr <1$.
Predicate $\true$ is always satisifed. 
We write $\vho\sat\f$ when the tuple of multisets $\vho$ satisfies the
conjunction of atomic predicates $\f$.

Once a process $p$ receives a multiset $\ho_p$, it uses it to do an  update of one of
its variables.
For this it finds the first conditional that $\ho_p$ satisfies and performs the
operation from the corresponding assignment. 

A \emph{condition} is a conjunction of atomic conditions: $\uni$, $\mult$,
$|\ho|>\thr\cdot |\Pi|$. 
A multiset $\ho$ satisfies $\uni$ when it contains just one value; it
satisfies $\mult$ if it contains more than one value.
A multiset $\ho$ satisfies $|\ho|>\thr\cdot |\Pi|$ when the size of $\ho$ is bigger
than $\thr\cdot n$, where $n$ is the number of processes. 
Observe that only predicates of the last type take into account possible
repetitions of the same value.

We can now define the \emph{update value} $\update_i(\ho)$, describing to which value
the process sets its variable in round $i$ upon receiving the multiset $\ho$. 
For this the process finds the first conditional statement in the sequence of instructions for
round $i$ whose condition is satisfied by $\ho-\set{?}$ and looks at the operation in
the statement:
\begin{itemize}
	\item if it is $x:=\min(\ho)$ then $\update_i(\ho)$ is the minimal value in $\ho-\set{?}$;
	\item if it is $x:=\smor(\ho)$ then $\update_i(\ho)$ is the smallest most
	frequent value in $\ho-\set{?}$;
		\item if no condition is satisfied then $\update_i(\ho)=?$.
\end{itemize}

A transition~\label{page:round-transition} $f\lact{\f}_i f'$ is possible when there exists a tuple of multisets
$(\ho_1,\dots,\ho_n)\sat\f$ such that for  all $p=1,\dots,n$: 
$\ho_p\incl \mset(f)$,  and
$f'(p)=\update_i(\ho_p)$.
Observe that $?$ value in $\ho_p$ is ignored by the $\update$ function, but not
by the communication predicate.

Finally, a transition $(f,d)\act{\p} (f',d')$, for $\p=(\f_1,\dots,\f_n)$, is
possible when there  is a sequence
\begin{equation*}
	f_0\lact{\f_1}_1f_1\lact{\f_2}_2\cdots\lact{\f_{r-1}}_{r-1}f_{r-1} \lact{\f_r}_r f_r
\end{equation*}
with:
\begin{itemize}
	\item $f_0=f$;
	\item $f'(p)=f_{\fo}(p)$ if $f_{\fo}(p)\not=?$, and $f'(p)=f(p)$ otherwise;
	\item $d'(p)=d(p)$ if $d(p)\not=?$, and $d'(p)=f_r(p)$ otherwise.
\end{itemize}
This means that $\inp$ is updated with the value from the input updating round
$\fo$, but only if the update is not $?$.
The value of $\dec$ cannot be updated, it can only be set if it has not
been set before.
For setting the value of $\dec$, the value from the last round is used. 

An \emph{execution} is a sequence of phase transitions.
An \emph{execution of an algorithm respecting a communication predicate} 
$(\lG\gp)\land(\lF(\p^1\land\lF(\p^2\land\dots (\lF\p^k)\dots)))$ is an 
infinite sequence:
\begin{equation*}
	(f_0,d_0)\act{\gp}^* (f_1,d_1)\act{\p\land\p^1}(f'_1,d'_1)\cdots \act{\gp}^*(f_k,d_k)\act{\p\land\p^k}(f'_k,d'_k)\act{\gp}^\w\cdots 
\end{equation*}
where $\act{\gp}^*$ stands for a finite sequence of $\act{\gp}$ transitions, and
$\act{\gp}^\w$ for an infinite sequence. 
For every execution there is some fixed $n$ standing for the number of processes, $f_0$ is any
$n$-tuple of values without $?$, and $d_0$ is the $n$-tuple of $?$ values. 
Observe that the size of the first tuple determines the size of every other
tuple. 
There is always a transition from every  configuration, so an execution cannot block. 

\begin{definition}[Consensus problem]\label{def:consensus}
An algorithm has \emph{agreement property} if for every number of processes $n$,
and for every state $(f,d)$ reachable by an execution of the algorithm, for all
processes $p_1$ and $p_2$, either $d(p_1)=d(p_2)$ or one of the two values is
$?$. 
An algorithm  has \emph{termination property} if for every $n$, and for every
execution there is a state $(f,d)$ on this execution with $d(p)\not=?$ for all
$p=1,\dots,n$. 
An algorithm \emph{solves consensus} if it has agreement and termination
properties. 
\end{definition}

\begin{remark}
Normally, the consensus problem also requires irrevocability and integrity
properties, but these are always guaranteed by the semantics: once set, a
process cannot change its $\dec$ value, and a variable can be set only to one of
the values that has been received. 
\end{remark}

\begin{remark}
The original definition of the Heard-Of model is open ended: it does not limit
possible forms of a communication predicate, conditions, or operations. 
Clearly, for the kind of result we present here, we need to fix them.
The original semantics uses process identifiers. 
We do not need them for the set of operations we consider here. 
\end{remark}

\begin{remark}
In the original definition processes are allowed to have identifiers. 
We do not need them for the set of operations we consider. 
Later we will add coordinators without referring to identifiers. 
This is a relatively standard way of avoiding identifiers while having
reasonable expressivity.
\end{remark}

\section{A characterization for the core language}\label{sec:core}

We present a characterization of all the algorithms in our language that solve
consensus. 
In later sections we will extend it to include timestamps and coordinators.
As it will turn out, for our analysis we will need to consider only two values
$a,b$ with a fixed order between them: we take $a$ smaller than $b$. 
This order influences the semantics of instructions: the result of $\min$
is $a$ on a multiset containing at least one $a$; the result of $\smor$ 
is $a$  on a multiset with the same number of $a$'s and $b$'s. 
Because of this asymmetry we mostly focus on the number of $b$'s in a tuple.
In our analysis we will consider tuples of the form $\bias(\th)$ for $\th<1$, i.e., a tuple
where we have $n$ processes (for some large enough $n$), out of which $\th\cdot n$ of them have
their value set to $b$; and the remaining ones to $a$. 
The tuple containing only $b$'s (resp.\ only $a$'s) is called $\solo$
(resp.\ $\soloa$).

We show that there is essentially one way to solve consensus. 
The text of the algorithm together with the form of the global predicate determines a
threshold $\bthr$.
We prove that in the language we consider here, there should be a \emph{unifier phase} which guarantees
that the tuple of $\inp$ values after the phase belongs to one of the following four
types: $\solo$, $\soloa$, $\bias(\th)$, or $\bias(1-\th)$ where $\th \ge \bthr$.
Intuitively, this means that there is a dominant value in the tuple. 
This phase should be followed by a \emph{decider phase} which
guarantees that if the tuple of $\inp$ is of one of the above mentioned types,
then all the processes decide.  While this ensures termination,
agreement is ensured by proving that some simple structural properties on 
the algorithm should always hold. 
In the rest of this section we give some observations and definitions in order to
state the result formally.



Before stating the characterization, we will make some observations that allow
us to simplify the structure of an algorithm, and in consequence simplify the
statements.

It is easy to see that in our languge we can assume that the list of
conditional instructions in each round can have at most one $\uni$ conditional
followed by a sequence of $\mult$ conditionals with non-increasing thresholds:
\begin{align*}
&\iif\ \uni(\ho)\land |\ho|>\thru^i\cdot |\Pi|\ \tthen\ x:=\op_u^i(\ho)\\
&\iif\ \mult(\ho)\land |\ho|>\thrm^{i,1}\cdot |\Pi|\ \tthen\ x:=\op_m^i(\ho)\\
& \vdots\\
&\iif\ \mult(\ho)\land |\ho|>\thrm^{i,k}\cdot |\Pi|\ \tthen\ x:=\op_m^i(\ho)
\end{align*}
We use superscript $i$ to denote the round number: so $\thruo$ is a threshold
associated to $\uni$ instruction in the first round, etc.
If round $i$ does not have a $\uni$ instruction, then $\thruo$ will be $-1$.  
For the sake of brevity, $\thrm^{i,k}$ will always denote the minimal 
threshold appearing in any of the $\mult$ instructions in round $i$ and 
$-1$ if no $\mult$ instructions exist in round $i$.

We fix a \emph{communication predicate}:
\begin{equation*}
	(\lG\gp)\land(\lF(\p^1\land\lF(\p^2\land\dots (\lF\p^k)\dots)))
\end{equation*}
\label{assumptions-com-predicate}Without loss of generality we can assume that every sporadic predicate
implies the global predicate; in consequence,  $\gp\land \p^i$ is equivalent to $\p^i$.
Recall that each of $\gp,\p^1,\dots,\p^k$ is an $r$-tuple of conjunctions of atomic predicates. 
We write $\p\dar_i$ for the $i$-th element of the tuple 
and so $\p$ is $(\p\dar_1,\dots,\p\dar_r)$.
By $\thr_i(\p)$ we denote the threshold constant appearing in 
the predicate $\p\dar_i$, i.e., if $\p\dar_i$ has $\f_{thr}$ as a conjunct,
then $\thr_i(\p) = thr$, if it has no such conjunct then $\thr_i(\p) = -1$ just
to avoid treating this case separately. 
We call $\p\dar_i$ an \emph{equalizer} if it has $\f_=$ as a conjunct.
In this case we also say that $\p$ has an equalizer. 

Recall (cf.~page~\ref{page:round-transition}) that a transition $f\lact{\p}_i
f'$ for a round  $i$ under a phase predicate $\p$ is possible when
there is a tuple of multisets $(\ho_1,\dots,\ho_n)\sat\p\dar_i$ such that for
all $p=1,\dots,n$: $\ho_p\in\mset(f)$ and $f'(p)=\update_i(\ho_p)$.

\begin{definition}
	We write $d\in\fire_i(f,\p)$ if there 
	there is $f'$ such that  $f\lact{\f}_i f'$ and $d=f'(p)$ for some $p$.
\end{definition}

\begin{definition}\label{def:preserving-and-solo-safe}
	 A round $i$ is \emph{preserving} w.r.t.\ $\p$ iff one of the three conditions
	hold: (i) it does not have an $\uni$ instruction, (ii) it does not have a $\mult$
	instruction, or (iii) $\thr_i(\p)< 	\max(\thru^i,\thrm^{i,k})$. 
	Otherwise the round is \emph{non-preserving}. 
	The round is \emph{solo safe} w.r.t.\ $\p$ if 	$0\leq \thru^i\leq \thr_i(\p)$.
\end{definition}

If $i$ is a preserving round, then there exists a tuple $f$ having no $?$ value,
such that we can produce $?$ out of $f$ after round $i$, this
allows us to not update $\inp$ in the phase with such a round, i.e., to preserve
the old values.
If $i$ is a non-preserving round no such tuple exists.
A solo safe round cannot alter the $\solo$ state.
These two properties are stated formally in the next lemma that follows directly
from the definitions.

\begin{lemma}\label{lem:preserving}
	A round $i$ is preserving w.r.t.\ $\p$ iff there is a tuple $f$ 
	such that $? \notin \dom(f)$ and $?\in \fire_i(f,\p)$.
	If a round $i$ is solo safe and $\solo\lact{\p}_i f$ then $f$ is $\solo$. 
\end{lemma}

\begin{remark} 
	Given a global predicate $\gp$ we can remove $\mult$
	instructions that will never be executed because there is an instruction with 
	a bigger threshold that is bound to be executed. 
	To see this, suppose rounds $1,\dots,i-1$ are non-preserving under $\gp$. 
	By Lemma~\ref{lem:preserving}, if $f \lact{\gp\ \dar_1}_1 f_1 \lact{\gp\
	\dar_2}_2 \dots \lact{\gp\ \dar_{i-1}}_{i-1} f_{i-1}$ and $f$ contains no $?$
	then $f_{i-1}$ contains no $?$ as well.
	Hence, no heard-of multi-set $\ho$ constructed from $f_{i-1}$ can have $?$ value.
	Consequently, if round $i$ is such that, say, 
	$\thr_i(\gp) > \thrm^{i,2}$ then we can be sure that 
	only the first two $\mult$ instructions in round $i$ can be executed
	under the predicate $\gp$: a process will always receive at least
	$\thrm^{i,2}$ fraction of values, and as there will be no $?$ value among them 
	the second threshold constraint will be satisfied.
	This implies that we can adopt the following assumption.
\end{remark}
	
\begin{assumption}
\label{assumption}
For every round $i$, if rounds $1,\dots,i-1$
are non-preserving under $\gp$ then 
\begin{equation}\label{eq:syntactic-property}
	\begin{cases}
		\thru^i\ge\thr_i(\gp) &\qquad \text{if round $i$ has $\uni$ instruction}\\
		\thrm^{i,k}\ge\thr_i(\gp) &\qquad \text{if round $i$ has $\mult$ instruction}
	\end{cases}
\end{equation} 
\end{assumption}
We put some restrictions on the form of algorithms we consider in our
characterization. 
They greatly simplify the statements, and as we argue, are removing cases that
are not that interesting anyway.

\begin{proviso}\label{proviso} We adopt the following additional syntactic restrictions:
	\begin{itemize}
		\item We require that the global predicate does not have an
		equalizer. 
		
		\item We assume that there is no $\mult$ instruction in the round $\fo+1$. 
	\end{itemize}
\end{proviso}
Concerning the first of the above requirements, if the global predicate has an
equalizer then it is quite easy to construct an
algorithm for consensus because  equalizer guarantees that in a given round all
the processes receive the same value.  
The characterization below can be extended to this case but would require to
mention it separately in all the statements.
Concerning the second requirement. We
prove in Lemma~\ref{lem:no-mult-sequence} that if such a $\mult$ instruction
exists then either the
algorithm violates consensus, or  the instruction will never be fired in any
execution of the algorithm and so can be removed without making an algorithm
incorrect.



In order to state our characterization we need to give formal definitions
of concepts we have discussed at the beginning of the section.

\begin{definition}\label{def:border-threshold}
	The \emph{border threshold} is $\bthr=
	\max(1-\thruo,1-\thrmok/2)$.
\end{definition}
Observe that $\bthr>1/2$ as $\thrmok<1$.

\begin{definition}\label{def:predicates}
	A predicate $\p$ is a
	\begin{itemize}
		\item \emph{Decider},  if all rounds are solo safe w.r.t. $\p$
		\item \emph{Unifier}, if the three conditions hold:
		\begin{itemize}
			\item $\thr_1(\p)\geq \thrmok$ and either  $\thr_1(\p)\geq \thruo$ or
			$\thr_1(\p)\geq \bthr$,
			\item there exists $i$ such that $1 \le i \le \fo$ and
			$\p\dar_i$ is an equalizer,
			\item rounds $2,\dots,i$ are
			non-preserving w.r.t.\ $\p$ and rounds $i+1,\dots \fo$ are 	solo-safe w.r.t.\ 
			$\p$
		\end{itemize}
	\end{itemize} 
\end{definition}

Finally, we list some syntactic properties of algorithms that, as we will see later,
imply the agreement property.
\begin{definition}\label{def:structure}
	An algorithm is \emph{syntactically safe} when:
	\begin{enumerate}
		\item First round has a $\mult$ instruction.
		\item Every round has a $\uni$ instruction.
		\item In the first round the operation in every $\mult$ instruction is $\smor$.
		\item $\thrmok/2 \geq 1-\thru^{\fo+1}$, and $\thruo \geq 1-\thru^{\fo+1}$.
	\end{enumerate}
\end{definition}

Recall that $\p^1,\dots,\p^k$ are the set of sporadic predicates from the
communication predicate. 
Without loss of generality we can assume that there is at least one sporadic
predicate, at a degenerate case it is always possible to take a sporadic
predicate that is the same as the global predicate. 
With these definitions we can state our characterization:
\eject

\begin{theorem}\label{thm:core}
	Consider algorithms in the core language satisfying syntactic constraints from
	Assumption~\ref{assumption} and Proviso~\ref{proviso}. 
	An algorithm solves consensus iff it is syntactically safe according to
	Definition~\ref{def:structure}, and it satisfies the condition:
	\begin{description}
		\item[T]
		There is $i\leq j$ such that $\p^i$ is a unifier and $\p^j$ is a decider.
	\end{description}
\end{theorem}

A two value principle is a corollary from the proof of the above theorem: an algorithm solves
consensus iff it solves consensus for two values. 
Indeed, it turns out that it is enough to work with three values $a,b$, and $?$.
The proof considers separately safety and liveness aspects of the consensus problem.

\begin{lemma}
	An algorithm violating structural properties from
	Definition~\ref{def:structure} cannot solve consensus.
	An algorithm with the structural properties has the agreement property.
\end{lemma}

\begin{lemma}
	An algorithm with the structural properties from Definition~\ref{def:structure}
	has the termination property iff it satisfies condition T from
	Theorem~\ref{thm:core}.
\end{lemma}

\section{A characterization for algorithms with timestamps}

We extend our characterization to algorithms with timestamps. 
Now, variable $\inp$ stores not only the value but also a timestamp, that is the number
of the last phase at which $\inp$  was updated.
These timestamps are used in the first round, as a process considers only
values with the most recent timestamp. 
The syntax is the same as before except that we introduce a new operation, called
$\maxts$, that must be used in the first round and nowhere else. 
So the form of the first round becomes\label{shape-for-timestamps}:

\RestyleAlgo{plain}
\begin{algorithm}[H]
		\Send{$(\inp,ts)$}{
		\lIf{$\cond_1^1(\ho)$}{$x_1:=\maxts(\ho)$}
		\vdots
		\lIf{$\cond_1^l(\ho)$}{$x_1:=\maxts(\ho)$}
		}
\end{algorithm}

The semantics of transitions for rounds and phases needs to take into account timestamps.  
The semantics changes only for the first round; its form becomes
$(f,t)\lact{\f}f'$, where $t$ is a vector of timestamps ($n$-tuple of natural numbers).
Timestamps are ignored by communication predicates and conditions, but are used
in the update operation. 
The operation $\maxts(\ho)$ returns the smallest  among values with the most
recent timestamp in $\ho$.

The form of a phase transition changes to  $(f,t,d)\act{\p} (f',t',d')$.
Value $t(p)$ is the timestamp of the last update of $\inp$ of process $p$ (whose
value is $f(p)$). 
We do not need to keep timestamps for $d$ since the value of $\dec$ can be set
only once.
Phase transitions are defined as before, taking into account the above mentioned change for
the first round transition, and the fact that in the round $\fo$ when $\inp$ is
updated then so is its timestamp. 
Some examples of algorithms with timestamps are presented in
Section~\ref{sec:examples}. 

As in the case of the core language, without loss of generality we can assume
conditions from Assumption~\ref{assumption} from the assumption on
page~\pageref{assumption}.
Concerning Proviso~\ref{proviso} on
page~\pageref{proviso}, we assume almost the same conditions, but now the
second one refers to the round $\fo$ and not to the round $\fo+1$, and is a bit stronger.

\begin{proviso}\label{ts-proviso}
	 We adopt the following syntactic restrictions:
		\begin{itemize}
			\item We require that the global predicate does not have an
			equalizer. 
			\item We assume that there is no $\mult$ instruction in the round $\fo$,
			and that $\thru^{\fo} \ge 1/2$.
		\end{itemize}	 
\end{proviso}
We prove (Lemma~\ref{lem:ts-min}) that if these two assumptions do not
hold then either the 	algorithm violates consensus, or  
we can remove the $\mult$ instruction and increase $\thru^\fo$ without
making an algorithm incorrect.

Our characterization resembles the one for the core language. 
The structural conditions get slightly modified: the condition on constants
is weakened, and there is no need to talk about $\smor$ operations in the fist round.
\begin{definition}\label{def:ts-structure}
	An algorithm is \emph{syntactically t-safe} when:
	\begin{enumerate}
		\item Every round has a $\uni$ instruction.
		\item First round has a $\mult$ instruction.
		\item $\thrmok\geq 1-\thru^{\fo+1}$ and $\thruo \geq 1-\thru^{\fo+1}$.
	\end{enumerate}
\end{definition}

We consider the same shape of a communication predicate as in the case of the
core language:
\begin{equation*}
	(\lG\gp)\land(\lF(\p^1\land\lF(\p^2\land\dots (\lF\p^k)\dots)))
\end{equation*}
We also adopt the same straightforward simplifying assumptions about the predicate as on page~\pageref{assumptions-com-predicate}.

A characterization for the case with timestamps uses a stronger version of a
unifier that we define now. 
The intuition is that we do not have $\bthr$ constant because of $\maxts$
operations in the first round. 
In other words, the conditions are the same as  before but when taking $\bthr>1$. 
\begin{definition}\label{def:strong-unifier}
	A predicate $\p$ is a \emph{strong unifier} $\p$ if it is a unifier in a sense
	of Definition~\ref{def:predicates} and 	$\thruo\leq\thr_1(\p)$.
\end{definition}
Modulo the above two changes, the characterization stays the same.

\begin{theorem}\label{thm:ts}
	Consider algorithms in the language with timestamps satisfying syntactic constraints from
	Assumption~\ref{assumption} and Proviso~\ref{ts-proviso}. 
	An algorithm satisfies consensus iff it is
	syntactically t-safe according to the structural properties from
	Definition~\ref{def:ts-structure}, and it satisfies:
		\begin{description}
			\item[sT]
			There are $i\leq j$ such that $\p^i$ is a strong unifier and $\p^j$ is a decider.
		\end{description}
\end{theorem}

 
\section{A characterization for algorithms with coordinators}

We consider algorithms equipped with coordinators. 
The novelty is that we can now have rounds where there is a unique process that
receives values from other processes, as well as rounds where there is a unique
process that sends values to other processes. 
For this we extend the syntax by introducing a round type that can be:  $\every$,
$\lr$ (leader  receive),  or $\ls$ (leader-send):
\begin{itemize}
\item A round of type $\every$ behaves as before.
\item In a round of type $\lr$ only one arbitrarily selected process receives values.
\item In a round of type $\ls$, the process selected in the immediately preceding $\lr$ round
sends its value to all other processes. 
\end{itemize}
If an $\ls$ round is not preceded by an $\lr$ round then an arbitrarily chosen process sends its value. 
We assume that every $\lr$ round is immediately followed by an $\ls$ round, because
otherwise the $\lr$ round would be useless. 
We also assume that $\inp$ and $\dec$ are not updated during $\lr$ rounds, as
only one process is active in these rounds.

For $\ls$ rounds we introduce a new communication predicate. 
The predicate $\f_{\ls}$ says that the leader
successfully sends its message to everybody; it  makes sense only for $\ls$
rounds.

These extensions of the syntax are reflected in the semantics. 
For convenience we introduce two new names for tuples: $\oneb$ is a tuple where all the
entries are $?$ except for one  entry which is $b$; similarly for $\onea$. 
Abusing the notation we also write $\oneq$ for $\soloq$, namely the tuple
consisting only of $?$ values.

If $i$-th round is of type $\lr$, we have a transition $f\lact{\p}_i \oned$
for every $d\in\fire_i(f,\p)$.
In particular, if $?\in\fire_i(f,\p)$ then $f\lact{\f}_i\soloq$
is possible. 

Suppose $i$-th round is of type $\ls$. 
If $\p\dar_i$ contains $\f_{\ls}$ as a conjunct then
\begin{align*}
	\oned\lact{\p}_i&\solo^d && \text{if round $(i-1)$ is of type $\lr$}	\\
	f\lact{\p}_i&\solo^d\ \text{for $d\in \dom(f)$}&& \text{otherwise}
\end{align*}
When $\p\dar_i$ does not contain $\f_{\ls}$ then independently of the type of
the round $(i-1)$ we have $f\lact{\p}_i f'$ for every $d\in\dom(f)$ and $f'$
such that $\dom(f')\incl \set{d,?}$.

We consider the same shape of a communication predicate as in the case of the
core language:
\begin{equation*}
	(\lG\gp)\land(\lF(\p^1\land\lF(\p^2\land\dots (\lF\p^k)\dots)))
\end{equation*}
We also adopt the same straightforward simplifying assumptions about the
predicate as on page~\pageref{assumptions-com-predicate}.

The semantics allows us to adopt some more simplifying assumptions about the syntax
of the algorithm, and the form of the communication predicate.

\begin{assumption}\label{assumption-ls-lr}
We assume that $\ls$ rounds do not have a $\mult$ instruction.
Indeed, from the above semantics it follows that $\mult$ instruction is never used in
a round of type $\ls$.
It also does not make much sense to use $\f_{\ls}$ in rounds other than of type
$\ls$. 
So to shorten some definitions we require that $\f_{\ls}$ can appear only  in
communication predicates for $\ls$-rounds. 
For similar reasons we require that $\f_=$ predicate is not used in $\ls$-rounds.
As we have observed in the first paragraph, we can assume that neither round
$\fo$ nor the last round are of type $\lr$.
\end{assumption}

The notions of preserving and solo-safe rounds get
extended to incorporate the new syntax

\begin{definition}
	A round of type $\ls$ is \emph{c-solo-safe} w.r.t.\ $\p$  if $\p_i$ has
	$\f_{\ls}$ as a conjunct, it is \emph{c-preserving} otherwise. 
	A round of type other than $\ls$ is \emph{c-preserving} or \emph{c-solo-safe}
	w.r.t\ $\p$ if it is so  in the sense of
	Definition~\ref{def:preserving-and-solo-safe}.  
\end{definition}

\begin{definition}
	A \emph{c-equalizer} is a conjunction containing a term of the form $\f_=$ or $\f_{\ls}$.	
\end{definition}

\begin{proviso}\label{c-proviso}
	We assume the same proviso as on page~\pageref{proviso}, but using the concepts
of c-equalizers instead of equalizers. 
\end{proviso}
To justify the proviso we prove that $\mult$ instruction in round $\fop$ cannot
be useful; cf.\ Lemma~\ref{lem:co-no-mult-fo}.

Assumption on page~\pageref{assumption} is also updated to using the notion of
c-preserving instead of preserving. We restate it for convenience.

\begin{assumption}\label{co-assumption}
For every round $i$, if rounds $1,\dots,i-1$
are non-c-preserving under $\gp$ then 
\begin{equation}\label{eq:c-syntactic-property}
	\begin{cases}
		\thru^i\ge\thr_i(\gp) &\qquad \text{if round $i$ has $\uni$ instruction}\\
		\thrm^{i,k}\ge\thr_i(\gp) &\qquad \text{if round $i$ has $\mult$ instruction}
	\end{cases}
\end{equation} 
\end{assumption}

Finally, the above modifications imply modifications of terms from Definition~\ref{def:predicates}.

\begin{definition}\label{def:c-predicates}
	A predicate $\p$ is called a
	\begin{itemize}
		\item \emph{c-decider}, if all rounds are c-solo safe w.r.t.\ $\p$.
		\item \emph{c-unifier}, if 
		\begin{itemize}
			\item $\thr_1(\p)\geq \thrmok$ and either  $\thr_1(\p)\geq \thruo$ or
			$\thr_1(\p)\geq \bthr$,
			\item there exists $i$ such that $1 \le i \le \fo$ and
			$\p\dar_i$ is an c-equalizer,
			\item rounds $2,\dots,i$ are
			non-c-preserving w.r.t.\ $\p$ and rounds $i+1,\dots \fo$ are 	c-solo-safe
			w.r.t.\  $\p$.
		\end{itemize}
	\end{itemize}		
\end{definition}

With these modifications, we get an analog of
Theorem~\ref{thm:core} for the case with coordinators subject to the modified 
provisos as explained above.

\begin{theorem}\label{thm:coordinators}
	Consider algorithms in the language with timestamps satisfying syntactic constraints from
	Assumptions~\ref{assumption-ls-lr}, \ref{co-assumption} and Proviso~\ref{c-proviso}. 
	An algorithm satisfies consensus iff the first round and
	the $(\fo+1)^{th}$ round are not of type $\ls$, it is syntactically safe according
	to Definition~\ref{def:structure}, and it satisfies the condition:
	\begin{description}
		\item[cT]
		There are $i\leq j$ such that $\p^i$ is a c-unifier and $\p^j$ is a c-decider.
	\end{description}
\end{theorem}

\section{A characterization for algorithms with coordinators and timestamps}

Finally, we consider the extension of the core language with both coordinators
and with timestamps. 
Formally, we extend the coordinator model with timestamps in the same way we
have extended the core model. 
So now $\inp$ variables store pairs (value, timestamp), and all the instructions
in the first round are $\maxts$ (cf. page~\pageref{shape-for-timestamps}).

\begin{proviso}\label{ts-co-proviso}
	We assume the same proviso as for timestamps: Proviso~\ref{ts-proviso} on
page~\pageref{ts-proviso}, but using the notion of c-equalizer.
\end{proviso}
As in the previous cases we justify our proviso by showing that the algorithm
violating the second condition would not be correct or the condition could be
removed (Lemma~\ref{lem:co-ts-no-mult-fo}).

The characterization is a mix of conditions from timestamps and coordinator
cases. 

\begin{definition} 
	A predicate $\p$ is a strong c-unifier if it is a c-unifier (cf.~Definition~\ref{def:strong-unifier})
	and $\thru^1 \le \thr_1(\p)$.
\end{definition} 

\begin{theorem}\label{thm:ts-coordinators}
	Consider algorithms in the language with timestamps satisfying syntactic constraints from
	Assumptions~\ref{assumption-ls-lr}, \ref{co-assumption} and Proviso~\ref{ts-co-proviso}. 
	An algorithm satisfies consensus iff the first
	round and the $(\fo+1)^{th}$ round are not of type $\ls$, it has the 	structural properties from 	Definition~\ref{def:ts-structure}, and it satisfies:
		\begin{description}
			\item[scT] There are $i\leq j$ such that $\p^i$ is a strong c-unifier and $\p^j$ is a c-decider.
		\end{description}
\end{theorem}

\section{Examples}\label{sec:examples}

We apply the characterizations from the previous sections to some 
consensus algorithms studied in the literature, and their variants.
We show some modified versions of these algorithms, and some impossibility
results, that are easy to obtain thanks to our characterization. 

Finally, we a show an algorithm that is new as far as we can tell. 
It is obtained by eliminating timestamps from a version of Paxos algorithm, and
using bigger thresholds instead.
 
\subsection{Core language}

First, we can revisit the parametrized Algorithm~\ref{alg:one-third} from
page~\pageref{alg:one-third}. 
This is an algorithm in the core language, and it depends on two thresholds.
Theorem~\ref{thm:core} implies that it solves consensus iff $\thr_1/2\geq
1-\thr_2$. 
In case of $\thr_1=\thr_2=2/3$ we obtain the well known OneThird algorithm. 
But, for example, $\thr_1=1/2$ and $thr_2=3/4$ are also possible solutions for
this inequality.
So Algorithm~\ref{alg:one-third} solves consensus for these values of thresholds.

Because of the conditions on constants, $\thrmok/2\geq 1-\thru^{\fo+1}$ coming from
Definition~\ref{def:structure},  it is not 
possible to have an algorithm in the core language where all constants are at
most $1/2$. This answers a question from~\cite{charron-heard-distributed09} for
the language we consider here.

The above condition on constants is weakened to $\thrmok\geq 1-\thru^{\fo+1}$ when
we have timestamps. 
In this case indeed it is possible to use only $1/2$ thresholds.
An algorithm from~\cite{Mar:17} is discussed later in this section. 

We can go further with a parametrization of the OneThird algorithm.
The one below is a general form of an algorithm with at most one $\mult$
instruction and two phases. 

\RestyleAlgo{ruled}
\begin{algorithm}[H]\label{alg:one-third-more-parameters}
		\Send{$(\inp)$}{
		\lIf{$\uni(\ho) \land |\ho| > \thru^1 \cdot |\Pi|$}{$x_1:=\inp:=\smor(\ho)$}
		\lIf{$\mult(\ho) \land |\ho| > \thrm^1 \cdot |\Pi|$}{$x_1:=\inp:=\smor(\ho)$}
	}
	\Send{$x_1$}{
		\lIf{$\uni(\ho) \land |\ho| > \thru^2 \cdot |\Pi|$}{$\dec:=\smor(\ho)$}
	}
	\BlankLine
	\Cp{$\lF(\p^1 \land \lF\p^2)$}
	\caption{Parametrized OneThird algorithm~\cite{charron-heard-distributed09}, $\thru^1$, $\thrm^1, \thru^2$ are constants from $(0,1)$}
\end{algorithm}

Let us list all the constraints on the constants that would make this algorithm
solve consensus. 
Observe, that if we want an algorithm with 2 rounds, by structural constraints
from Definition~\ref{def:structure}, there must be $\mult$ instruction in the
first round and there cannot be $\mult$ instruction in the second round.
The operation in the $\mult$ instruction must be $\smor$.
Both rounds need to have $\uni$ instruction.

The structural constraints from Definition~\ref{def:structure} imply
\begin{equation*}
	\thrm^1/2\geq 1-\thru^2 \quad\text{and}\quad \thru^1\geq 1-\thru^2
\end{equation*}
Recall that the formula for border threshold is
\begin{equation*}
	\bthr= \max(1-\thruo,1-\thrmok/2)	
\end{equation*}

We will consider only the case when the global predicate is $(\true,\true)$ so
there are no constraints coming from the proviso. 
Let us see what can be $\p^1$ and $\p^2$ so that we have a unifier and a
decider.

Decider is a simpler one. 
We need to have $\p^2:=(\f_{\thru^1},\f_{\thru^2})$, or some bigger thresholds.

For a unifier we need $\p^1:=(\f_{=}\land\f_{\th},\true)$, but we may have a
choice for $\th$ with respect to the constants $\thru^1, \thrm^1, \thru^2$. 
\begin{itemize}
	\item Suppose $\thru^1\leq \thrm^1$. Then the constraints on unifier reduce to $\th\geq\thrm^1$.
	\item Suppose $\thru^1>\thrm^1$ and all the constrains for the algorithm to
	solve consensus are satisfied. Then we can
	decrease $\thru^1$ to $\thrm^1$, and they will be still satisfied. Actually,
	one can show that one can decrease $\thru^1$ to $\thrm^1/2$. 
\end{itemize}

To sum up, the constraints are $\th\geq\thrm^1$, $\thru^1=\thrm^1/2\geq 1
-\thru^2$. 
If we want to keep the constants as small as possible we take $\th=\thrm^1$.
We get the best constraints as a function of $\thrm^1$:
\begin{equation*}
	\p^1=(\f_{=}\land\f_{\thrm^1},\true)\qquad \p^2=(\thrm^1/2,1-\thrm^1)
\end{equation*}

\begin{remark}
	The notion of a unifier (Definition~\ref{def:predicates}) suggests that there
	are two types of algorithms for consensus.
	The first type has a round that guarantees that every process has the same
	value (a unifier with $\thruo\leq \thrmok$), and a later round that makes all the processes decide
	(decider).
	The second type has a weaker form of unifier ($\thruo> \thrmok$) that only
	guarantees bias between values to be above $\bthr$ (or below $1-\bthr$). 
	Then the decider is stronger and makes every process decide even if not all
	processes have the same value.
	
	We do not see algorithms of the second type in the literature, and indeed
	the characterization says why. 
	The second type appears when $\thruo> \thrmok$, but our characterization
	implies that in this case we can decrease $\thruo$ to $\thrmok/2$ and the
	algorithm will be still correct. 
	So unless there are some constrains external to the model, algorithms with a
	weaker form of unifier are not interesting. 
\end{remark}

\subsection{Timestamps}

We start with a timestamp algorithm from~\cite{Mar:17} that uses only $1/2$ thresholds.
 
 \begin{algorithm}[H]\label{alg:new-timestamp}
		 \Send{$(\inp,ts)$}{
		 \lIf{$\uni(\ho) \land |\ho| > 1/2 \cdot |\Pi|$}{$x_1:=\maxts(\ho)$}
		 \lIf{$\mult(\ho) \land |\ho| > 1/2 \cdot |\Pi|$}{$x_1:=\maxts(\ho)$}
	 }
	 \Send{$x_1$}{
		 \lIf{$\uni(\ho) \land |\ho| > 1/2 \cdot |\Pi|$}{$x_2:=\inp:=\smor(\ho)$}
	 }
	 \Send{$x_2$}{
		 \lIf{$\uni(\ho) \land |\ho| > 1/2 \cdot |\Pi|$}{$\dec:=\smor(\ho)$}
	 }
	 \BlankLine
	 \Cp{$\lF(\p^1)$ where $\p^1 := (\f_=\land\f_{1/2},\ \f_{1/2},\ \f_{1/2})$}
	 \caption{A timestamp algorithm from~\cite{Mar:17}}
 \end{algorithm}
 This algorithm is correct by Theorem~\ref{thm:ts}. 
 The theorem also says that the communication predicate can be weakened to
 $\lF(\p^1\land\lF\p^2)$ where $\p^1=(\f_=\land\f_{1/2},\ \f_{1/2},\ \true)$
 and $\p^2=(\f_{1/2},\ \f_{1/2},\ \f_{1/2})$.
 \medskip

 If we do not want to have $\f_=$ requirement on the first round where we check
 timestamps, we can consider the following modification of the above.
 
 \begin{algorithm}[H]\label{alg:mod-new-timestamp} 
		 \Send{$(\inp,ts)$}{
		 \lIf{$\uni(\ho) \land |\ho| > 1/2 \cdot |\Pi|$}{$x_1:=\maxts(\ho)$}
		 \lIf{$\mult(\ho) \land |\ho| > 1/2 \cdot |\Pi|$}{$x_1:=\maxts(\ho)$}
	 }
	 \Send{$x_1$}{
		 \lIf{$\uni(\ho) \land |\ho| > 1/2 \cdot |\Pi|$}{$x_2:=\smor(\ho)$}
		 \lIf{$\mult(\ho) \land |\ho| > 1/2 \cdot |\Pi|$}{$x_2:=\smor(\ho)$}
	 }
	 \Send{$x_2$}{
		 \lIf{$\uni(\ho) \land |\ho| > 1/2 \cdot |\Pi|$}{$x_3:=\inp:=\smor(\ho)$}
	 }
	 \Send{$x_3$}{
		 \lIf{$\uni(\ho) \land |\ho| > 1/2 \cdot |\Pi|$}{$\dec:=\smor(\ho)$}
	 }
	 \BlankLine
	 \Cp{$\lF(\p^1\land\lF\p^2)$}
	 \Where{$\p^1=(\f_{1/2},\ \f_=\land\f_{1/2}, \ \f_{1/2}, \ \true)$
	 and $\p^2=(\f_{1/2},\ \f_{1/2},\ \f_{1/2}, \ \f_{1/2})$}
	 \caption{A modification of a timestamp algorithm from~\cite{Mar:17}}
 \end{algorithm}
 Note that by Theorem~\ref{thm:ts}, when we move the equalizer to the
 second round, there necessarily has to be a $\mult$ instruction in the
 second round.
 
 \subsection{Timestamps and coordinators}
 When we have both timestamps and coordinators, we get variants of Paxos
 algorithm.
 
 \begin{algorithm}[H]\label{alg:paxos}
		\Send{$(\inp,ts)$ $\lr$}{
		\lIf{$\uni(\ho) \land |\ho| > 1/2 \cdot |\Pi|$}{$x_1:=\maxts(\ho)$}
		\lIf{$\mult(\ho) \land |\ho| > 1/2 \cdot |\Pi|$}{$x_1:=\maxts(\ho)$}
	}
	\Send{$x_1$ $\ls$}{
		\lIf{$\uni(\ho)$}{$x_2:=\inp:=\smor(\ho)$}
	}
	\Send{$x_2$ $\lr$}{
		\lIf{$\uni(\ho) \land |\ho| > 1/2 \cdot |\Pi|$}{$x_3:=\smor(\ho)$}
	}
	\Send{$x_3$ $\ls$}{
		\lIf{$\uni(\ho)$}{$dec:=\smor(\ho)$}
	}
	\BlankLine
	\Cp{$\lF(\p^1)$ where $\p^1 := (\f_{1/2},\ \f_{\ls},\
	\f_{1/2},\ \f_\ls)$} 
	\caption{Paxos algorithm}
\end{algorithm} 
The algorithm is correct by Theorem~\ref{thm:ts-coordinators}. 
One can observe that without modifying the code there is not much room for
improvement in this algorithm.
A decider phase is needed to solve consensus, and $\p_1$ is a minimal requirement for a
decider phase. 
A possible modification is to change the thresholds in the first round to, say, $1/3$
and in the third round to $2/3$ (both in the algorithm and in the communication predicate).

Chandra-Toueg algorithm in the Heard-Of model is actually syntactically the same
 as four round Paxos~\cite{Mar:17}. 
 The communication predicate is even stronger so it clearly satisfies our constraints.\\

 The next example is a three round version of Paxos algorithm. 
 
 \begin{algorithm}[H]\label{alg:paxos-three}
		 \Send{$(\inp,ts)$ $\lr$}{
		 \lIf{$\uni(\ho) \land |\ho| > 1/2 \cdot |\Pi|$}{$x_1:=\maxts(\ho)$}
		 \lIf{$\mult(\ho) \land |\ho| > 1/2 \cdot |\Pi|$}{$x_1:=\maxts(\ho)$}
	 }
	 \Send{$x_1$ $\ls$}{
		 \lIf{$\uni(\ho)$}{$x_2:=\inp:=\smor(\ho)$}
	 }
	 \Send{$x_2$ $\every$}{
		 \lIf{$\uni(\ho) \land |\ho| > 1/2 \cdot |\Pi|$}{$\dec:=\smor(\ho)$}
	 }
	 \BlankLine
	 \Cp{$\lF(\p^1)$ where $\p^1 := (\f_{1/2},\ \f_{\ls},\ \f_{1/2})$}
	 \caption{Three round Paxos algorithm}
 \end{algorithm}
 The algorithm is correct by Theorem~\ref{thm:ts-coordinators}. 
 Once again it is possible to change constants in the first round to $1/3$ and in
 the last round to $2/3$ (both in the algorithm and in the communication predicate).
 \bigskip
 
\subsection{Coordinators without timestamps}

 One can ask if it is possible to have an algorithm with coordinators without
 timestamps.
 Here is a possibility that resembles three round Paxos:
 
 \begin{algorithm}[H]\label{alg:paxos-coordinators}
		 \Send{$(\inp)$ $\lr$}{
		 \lIf{$\uni(\ho) \land |\ho| > 2/3 \cdot |\Pi|$}{$x_1:=\smor(\ho)$}
		 \lIf{$\mult(\ho) \land |\ho| > 2/3 \cdot |\Pi|$}{$x_1:=\smor(\ho)$}
	 }
	 \Send{$x_1$ $\ls$}{
		 \lIf{$\uni(\ho)$}{$x_2:=\inp:=\smor(\ho)$}
	 }
	 \Send{$x_2$ $\every$}{
		 \lIf{$\uni(\ho) \land |\ho| > 2/3 \cdot |\Pi|$}{$\dec:=\smor(\ho)$}
	 }
	 \BlankLine
	 \Cp{$\lF(\p)$\quad where $\p := (\f_{2/3},\f_{\ls},\f_{2/3})$}
			 \caption{Three round coordinator algorithm}
 \end{algorithm}
 The algorithm solves consensus by Theorem~\ref{thm:coordinators}. 
 The constants are bigger than in Paxos because we do not have timestamps: 
 the constraints on constants come from
 Definition~\ref{def:structure}, and not from Definition~\ref{def:ts-structure}.
 The advantage is that we do not need time-stamps, while keeping the same
 structure as for three-round Paxos.
 We can parametrize this algorithm in the same way as we did for
 Algorithm~\ref{alg:one-third-more-parameters}.

 

\section{Proof of the characterization for the core language}
In this section we prove Theorem~\ref{thm:core}, namely a characterization of
algorithms in the  core language that solve consensus.


We fix a \emph{communication predicate}
\begin{equation*}
	(\lG\gp)\land(\lF(\p^1\land\lF(\p^2\land\dots (\lF\p^k)\dots)))
\end{equation*}
Recall that each of $\gp,\p^1,\dots,\p^k$ is an $r$-tuple of atomic predicates. 
We write $\p\dar_i$ for the $i$-th element of the tuple. 
So $\p$ is $(\p\dar_1,\dots,\p\dar_r)$.
Often we will write $\p_i$ instead of $\p\dar_i$, in particular when $\p_i$
appears as a subscript; for example $f\lact{\p_i} f'$.
If $\f$ is a conjunction of atomic predicates, then by $\thr(\f)$ we denote
the threshold constant appearing in $\f$, i.e., if 
$\f$ has $\f_{thr}$ as a conjunct then $\thr(\f) = thr$, if it has 
no such conjunct then $\thr(\f) = -1$.
\medskip


\begin{definition}
	We define several tuples of values. 
	All these tuples will be $n$-tuples for some fixed but large enough $n$
	and will be over  $\set{a, b}$, or $\set{?,b}$ or $\set{?,a}$.
	For $\th < 1$, the tuple $\bias(\th)$ is a tuple containing only $a$'s and
	$b$'s with $|b|=\th\cdot n$. 
	Tuple $\bias(1/2)$ is also called $\spread$ to emphasize that there is the
	same number of $a$'s and $b$'s.
	A tuple consisting only of $b$'s is denoted $\solo$.
	Similarly, $\biasq(\th)$ is a tuple over $\set{?,b}$ with $|b|=\thr\cdot
	n$ and $\biasq_a(\th)$ is a tuple over $\set{?,a}$ with $|a| = \thr \cdot n$. We also write $\solo^?$ for a tuple consisting only of $?$'s. 
	Finally, we write $\solo^a$ for a tuple consisting only of $a$'s.
\end{definition}

\noindent\textbf{Notations:}
\begin{itemize}
	\item For a tuple of values $f$ and a predicate $\p$ we write  $\fire_i(f,\p)$
	instead of $\fire_i(f,\p\dar_i)$. 
	Similarly we write $\thr_i(\p)$ for $\thr(\p\dar_i)$.
	\item If $f, f'$ are tuples of values, we write $f\act{\p}f'$ instead of
	$(f,\soloq)\act{\p}(f',\soloq)$. 
\end{itemize}



Recall that the border threshold for an algorithm,
by Definition~\ref{def:border-threshold}, is
\begin{equation*}
\bthr= 	\max(1-\thruo,1-\thrmok/2)
\end{equation*}
Observe that $\bthr>1/2$ as $\thrmok<1$.

The proof of Theorem~\ref{thm:core} is divided into three parts. 
First we show that if an algorithm does not satisfy the structural properties
then it violates agreement.
Then  we restrict our attention to algorithms with the structural properties. 
We show that if condition T holds, then consensus is satisfied. 
Finally, we prove that if condition T does not hold then the algorithm does not
have the termination property. 

To simplify the statements of the lemmas, we adopt the following convention. If
some condition is proved as necessary for consensus, then for the forthcoming
lemmas, that condition is assumed. For example, in Lemma~\ref{lem:no-uni}, we prove that all
rounds should have a $\uni$ instruction. Hence after Lemma~\ref{lem:no-uni}, it it implicitly
assumed that all algorithms considered have a $\uni$ instruction in every round.

\subsubsection*{Part 1: Structural properties}

\begin{lemma}\label{lem:no-mult}
	If no $\mult$ instruction is present in the first round then the  algorithm may not terminate.
\end{lemma}
\begin{proof}
	Suppose no $\mult$ instruction is present in the first round.
	It is easy to verify that for every predicate $\p$, we have
	$\spread \lact{\p_1}_1 \solo^?$ and $\solo^? \lact{\p_i}_i \solo^?$ for $i > 1$. Hence we have the phase transition $\spread \act{\p} \spread$.
\end{proof}

\begin{lemma}\label{lem:no-uni}
	If there is a round without a $\uni$ instruction then the algorithm does not terminate.
\end{lemma}
\begin{proof}
	Let $i$ be the round without a $\uni$ instruction. It is easy
	to verify that for every predicate $\p$, we have $\solo \lact{\p_j}_j \solo$ for $j < i$, $\solo \lact{\p_i}_i \solo^?$ and $\solo^? \lact{\p_j}_j \solo^?$ for $j > i$. Hence we get the phase
	transition $\solo \act{\p} \solo$.
\end{proof}

Before considering the remaining structural requirements we state some useful
lemmas.


\begin{lemma}\label{lem:spread}
	Suppose all $\mult$ instructions in the first round have $\smor$ as the
	operation. Then for every predicate $\p$ we have 	$\set{a,b}\incl\fire_1(\spread,\p)$. 
\end{lemma}
\begin{proof}
	From $\spread$, it is easy to see that we can construct a multiset $\ho$ containing more $a$'s
	than $b$'s such that the size of $\ho$ is bigger than $\thr_1(\p)$ and
	$\thrm^{1,1}$.
	Similarly we can construct a multiset having more $b$'s than $a$'s.
	This then implies that $\set{a,b} \incl \fire_1(\spread,\p)$.
	
\end{proof}

\begin{lemma}\label{lem:bias-preserving}
	If a round $i$ is preserving w.r.t.\ $\p$ then
	$\set{b,?}\in\fire_i(\bias(\th),\p)$ for all sufficiently big $\th$.
	Similarly $\set{a,?} \in \fire_i(\bias(\th),\p)$ for all sufficiently small $\th$. 
\end{lemma}
\begin{proof}
	Let $\th > \max(\thru^i,\thr_i(\p))$.
	Because of the $\uni$ instruction, it is then clear that $b \in \fire_i(\bias(\th),\p)$.
	Since the round is preserving (and since $\uni$ instructions are present in every round), 
	either there is  no $\mult$ instruction in round $i$ or  $\thr_i(\p)<\thru^i$, or
	$\thr_i(\p)<\thrm^{i,k}$. 
	In the first case, let $\ho$ be the entire tuple.
	In the second case, let $\ho$ be a multi-set consisting only of $b$'s but of size smaller than $\thru^i$ and bigger than $\thr_i(\p)$.
	In the third case, let $\ho$ be a multi-set of size smaller
	than $\thrm^{i,k}$ (and bigger than $\thr_i(\p)$) with at least one  $a$, and one $b$. In all the cases, it is clear 
	that $? = \update_i(\ho)$ and so $? \in \fire_i(\bias(\th),\p)$.
	We can argue similarly for the other case as well.
\end{proof}

\begin{lemma}\label{lem:ab-if-mult}
	For every predicate $\p$, for every round $i$ with $\mult$ instruction, there is
	a threshold $\th\geq 1/2$ such that $\set{a,b}\incl \fire_i(\bias(\th),\p)$.
\end{lemma}
\begin{proof}
	Let $I$ be the $\mult$ instruction in round $i$ with the biggest 
	threshold. 
	This threshold is called $\thrm^{i,1}$ in our notation. 
	If the operation of $I$ is $\smor$ then we take $\th=1/2$ and argue
	similar to the proof of Lemma~\ref{lem:spread}.
	If the operation of $I$ is $\min$ then we take $\th>\max(\thr_i(\p), \thruo,1/2)$.
	Because of the $\uni$ instruction, we can 
	get $b$ by sending a multi-set $\ho$ consisting of all the $b$'s in $\bias(\th)$. 
	Further because of the instruction $I$, if we send the entire tuple as a
	multi-set, we get $a$. 
\end{proof}

\begin{lemma}\label{lem:a-if-min}
	Suppose the first round has a  $\mult$ instruction with
	$\min$ as operation. Then $a\in \fire_1(\bias(\th),\gp)$ for 
	every $\th > 0$.
\end{lemma}
\begin{proof}
	Let the $j^{th}$ $\mult$ instruction be the instruction
	with the $\min$ operation. Let $\ho$ be any multiset containing
	at least one $a$ and one $b$ and is of size just above $\thrm^{1,j}$.
	By observation~(\ref{eq:syntactic-property}) we have
	$\thrm^{1,j} \ge \thr_1(\gp)$ and so we have that $\ho \models \gp_1$ and
	$a  = \update_1(\ho)$. 
\end{proof}

The next sequence of lemmas tells us what can happen in a sequence of rounds.

\begin{lemma}\label{lem:any-qbias-from-qbias}
	Suppose none of $\p_k,\dots,\p_l$ is an equalizer. If $\set{b,?}\incl
	\fire_k(f,\p_k)$ then for every $\th'$ we have
	$f\lact{\p_k}_k\dots\lact{\p_l}_l \biasq(\th')$. Similarly, for $b$ replaced
	by $a$, and $\biasq(\th')$ replaced by $\biasqa(\th')$.
\end{lemma}
\begin{proof}
	The proof is by induction on $k-l$. 
	If $k=l$ then the lemma is clearly true since we can produce both $b$ and $?$
	values, and $\p_k$ is not an equalizer.
	For the induction step, consider the last round $l$, and let $\th''=\thru^l+\e$ for
	some small $\e$. 
	We have $b\in\fire_l(\biasq(\th''),\p_l)$ because of the $\uni$ instruction.
	We can also construct a multiset $\ho$ from $\biasq(\th'')$ 
	of size $1-\e' > \thr(\p_l)$ for some small $\e' > \e$
	containing $\th''-\e'$ fraction of $b$'s and $1-\th''$ fraction of
	$?$.
	This multiset shows that  $?\in\fire_l(\biasq(\th),\p_l)$.
	So from $\biasq(\th'')$ in round $l$ we can get $\biasq(\th')$ for any $\th'$.
	The induction assumption gives us $f\lact{\p_k}_k\dots\lact{\p_{l-1}}_{l-1}
	\biasq(\th'')$, so we are done.
\end{proof}

\begin{lemma}\label{lem:any-frequency}
	Suppose none of $\p_k,\dots,\p_l$ is an equalizer, and all rounds $k\dots l$
	have $\mult$ instructions.
	If $\dom(f')\incl \fire_k(f,\p_k)$ and $? \notin \dom(f')$
	then $f\lact{\p_k}_k\dots\lact{\p_l}_l f'$ is possible.
\end{lemma}
\begin{proof}
	We proceed by induction on $l-k$.
	The lemma is clear when $k = l$. Suppose $k \neq l$. 
	Consider three cases:
	
	Suppose $f'$ is $\solo^a$ or $\solo$. By induction hypothesis
	we can reach $f'$ after round $l-1$. Since round $l$ 
	has a $\uni$ instruction it is clear that $f' \lact{\p_l}_l f'$.
	
	Suppose $a,b \in \dom(f')$. 
	Lemma~\ref{lem:ab-if-mult} says that there is $\th$ for which
	$\set{a,b}\in\fire_l(\bias(\th),\p_l)$. 
	Hence $\bias(\th)\lact{\p_l}_l f'$. 
	By induction hypothesis we can reach $\bias(\th)$ after round $l-1$.
\end{proof}

\begin{lemma}\label{lem:any-qbias-from-bias}
	Suppose none of $\p_k,\dots,\p_l$ is an equalizer, and some round $k,\dots,l$
	does not have a $\mult$ instruction. 
	For every $\th$ and every $f$  such that $\set{a,b}\in
	\fire_k(f,\p_k)$ we have 	$f\lact{\p_k}_k\dots\lact{\p_l}_l \biasq(\th)$, 
	and $f\lact{\p_k}_k\dots\lact{\p_l}_l \biasqa(\th)$.
\end{lemma}
\begin{proof}
	Let $i$ be the first round without $\mult$ instruction.
	Using Lemma~\ref{lem:any-frequency}, from the tuple $f$ at round $k$, we can arrive at round $i$
	with the tuple $\bias(\th)$ for any $\th$. 
	We  choose $\th$ according to Lemma~\ref{lem:bias-preserving} so that 
	$\set{b,?}\incl \fire_i(\bias(\th),\p_i)$. 
	Then we can apply Lemma~\ref{lem:any-qbias-from-qbias} to prove the claim.
	The reasoning for $\biasqa$ is analogous.
\end{proof}

\begin{lemma}\label{lem:no-mult-sequence}
	If round $\fo+1$ contains a $\mult$ instruction  then the algorithm 
	does not satisfy agreement, or the $\mult$ instruction can be removed without altering the
	correctness of the algorithm.
\end{lemma}

\begin{proof}
	Suppose round $\fo+1$ contains a $\mult$ instruction
	
	The first case is when there does not exist any tuple $f$ and an execution
	$f \lact{\gp_1}_1 f_1 \dots \lact{\gp_{\fo-1}}_{\fo-1} f_{\fo-1}$
	such that $a,b \in \fire_\fo(f_{\fo-1},\gp)$. It is then clear
	that the $\mult$ instructions in round $\fo+1$ will never be fired
	and so we can remove all these instructions in round $\fo+1$.
	
	So it remains to examine the case when  there exists a tuple $f$ with
	$f \lact{\gp_1}_1  f_1\lact{\gp_2} \cdots\lact{\gp_{\fo-1}} f_{\fo-1}$ such that
	$a,b \in \fire_{\fo}(f_{\fo-1},\gp)$. 
	In this case we get
	$f_{\fo-1} \lact{\gp_{\fo}}_{\fo} \bias(\th)$ for arbitrary $\th$.
	Let $I$ be the $\mult$ instruction in round $\fo+1$ with the
	highest threshold value.
	Recall that, by proviso from page~\ref{proviso}, $\gp$ is not an equalizer.
	We consider two sub-cases: 
	
	Suppose $I$ has $\smor$ as its operation.
	Then we consider 	$f_{\fo-1} \lact{\gp_{\fo}}_{\fo} \spread$. 
	As $I$ has $\smor$ as operation, from $\spread$ we can construct a multiset $\ho$ containing more $a$'s
	than $b$'s such that the size of $\ho$ is bigger than $\thr_{\fo+1}(\gp)$ and $\thrm^{\fo+1,1}$.
	Similarly we can construct a multiset having more $b$'s than $a$'s.	
	Hence we get $\spread \lact{\gp_{\fo+1}}_{\fo+1} \bias(\th')$
	for arbitrary $\th'$. 
	If all rounds after $\fo+1$ have $\mult$ instructions, then we can
	apply Lemma~\ref{lem:any-frequency} to conclude that we can 
	reach the tuple $\spread$ after round $r$, thereby deciding 
	on both $a$ and $b$ and violating agreement.
	Otherwise we can use Lemma~\ref{lem:any-qbias-from-bias} to conclude that we can 
	reach the tuple $\spread^?$ after round $r$ and hence
	make half the processes decide on $b$. 
	Notice that after this phase  the state of the algorithm is $(\spread,\spread^?)$.
	We know, by Lemma~\ref{lem:no-mult} that the first round has a $\mult$ instruction.
	This instruction has $\smor$ or $\min$ as its operation,
	it is clear that in either case, $a \in \fire_1(\spread,\gp)$ and
	so we can get $\spread \lact{\gp_1}_1 \solo^a$ and
	$\solo^a \lact{\gp_i}_i \solo^a$ for $i > 1$, thereby making the
	rest of the undecided processes decide on $a$. Hence
	agreement is violated.
	
	Suppose $I$ has $\min$ as its operation. Then we consider
	$f_{\fo-1} \lact{\gp_{\fo}}_{\fo} \bias(\th)$ where $\th > \max(\thru^{\fo+1},\thru^{1},\thr_{\fo+1}(\gp	))$ is sufficiently big.
	It is clear that $b \in \fire_{\fo+1}(\bias(\th),\gp)$.
	Further if we take our multi-set $\ho$ to be $\bias(\th)$ itself,
	then (because of the instruction $I$) we have $a \in \fire_{\fo+1}(\bias(\th),\gp)$.
	Hence we get $\bias(\th) \lact{\gp_{\fo+1}}_{\fo+1} \bias(\th')$
	for arbitrary $\th'$. 
	As in the previous case, either this immediately allows us to conclude that agreement is violated, or this allows us to make half the processes decide on $a$. In the latter case, 
	note that the state of the algorithm after this 
	phase will be $(\bias(\th),\spread^?_a)$.
	Since $\th \ge \thru^1$ and since $\thru^1 \ge \thr_1(\gp)$ by observation~(\ref{eq:syntactic-property}), it follows
	that $b \in \fire_1(\bias(\th),\gp)$.
	Hence we can get $\solo$ as the tuple after the first round and decide
	on $b$, as in the previous case.

	
	
\end{proof}

\begin{lemma}\label{lem:no-min}
	If the first round has $\mult$ instruction with
	$\min$ as the operation then the  algorithm does not satisfy agreement. 
\end{lemma}

\begin{proof}
	Suppose that indeed the first round does have a $\mult$ instruction with
	$\min$ operation. 
	Thanks to our proviso, the global predicate  does not  have an equalizer, hence we
	can freely apply Lemmas~\ref{lem:any-frequency} and~\ref{lem:any-qbias-from-bias}. 
	
	We use Lemma~\ref{lem:ab-if-mult} to find $\th$ with $\set{a,b}\incl
	\fire_1(\bias(\th),\gp)$. 
	We consider two cases.
	
	If all the rounds $2,\dots,\fo$ have a $\mult$  instruction then
	Lemma~\ref{lem:any-frequency} allows us to get $\bias(\th')$, for arbitrary
	$\th'$,  after round 	$\fo$.
	By Lemma~\ref{lem:no-mult-sequence}  there is no $\mult$ instruction in round
	$\fo+1$. 
	By Lemma~\ref{lem:bias-preserving} there is $\th'$ such that $\set{b,?}\incl
	\fire_{\fo+1}(\bias(\th'),\gp)$. 
	Using Lemma~\ref{lem:any-qbias-from-qbias} we can make some process decide on $b$, while
	keeping the other processes undecided. 
	Hence the state of the algorithm after this phase is $(bias(\th'),\spread^?)$. 
	By Lemma~\ref{lem:a-if-min}, $a\in\fire_1(\bias(\th'),\gp)$, and so we can get
	$\solo^a$ as the tuple after the first round and make all the other processes decide on $a$. 
	
	The second case is when one of the rounds $2,\dots,\fo$ does not have a $\mult$ 
	instruction. 
	For arbitrary $\th'$, Lemma~\ref{lem:any-qbias-from-bias} allows us to get
	$\biasq(\th')$ after round $\fo$.
	As in the above case, we use it to decide on $b$ for some process while
	leaving other undecided. 
	In the next phase we make other processes decide on $a$.
\end{proof}

\begin{lemma}\label{lem:constants}
	If property of constants from Definition~\ref{def:structure} is not satisfied then the algorithm does not satisfy
	agreement.
\end{lemma}

\begin{proof}
	We consider an execution of a phase under the
	global predicate and so we can freely use Lemmas~\ref{lem:any-frequency}
	and~\ref{lem:any-qbias-from-bias}. 
	We have seen in Lemma~\ref{lem:no-min} that in the first round all the
	$\mult$ instructions must be $\smor$. We start with the state
	$(\spread,\solo^?)$. 
	
	We consider two cases. 
	
	\textbf{First case :} There are no preserving rounds before
	round $\fo+1$. Hence every round before $\fo+1$ has a $\mult$ instruction. By Lemma~\ref{lem:any-frequency} from $\spread$
	we can get $\bias(\th)$ (for any $\th$) as the tuple before
	round $\fo+1$. Choose $\th = \thru^{\fo+1} + \e$ for some
	small $\e$. 
	By Lemma~\ref{lem:no-mult-sequence} we know 
	that round $\fo+1$ does not have any $\mult$ instruction.
	This implies that $? \in \fire_{\fo+1}(\bias(\th),\gp)$.
	Further, by observation~(\ref{eq:syntactic-property}) we know
	that $\thru^{\fo+1} \ge \thr_{\fo+1}(\gp)$. 
	Therefore, $b \in \fire_{\fo+1}(\bias(\th),\gp)$. Hence
	$\set{b,?} \subseteq \fire_{\fo+1}(\bias(\th),\gp)$.
	
	\textbf{Second case: } There is a round $j < \fo+1$ such
	that round $j$ is preserving. Let $j$ be the first such round.
	By Lemma~\ref{lem:any-frequency} from $\spread$ we can get
	$\bias(\thru^j+\e')$ (for some small $\e'$) before round $j$.
	Since round $j$ is preserving it follows that either round $j$ has no $\mult$ instruction or $\thr_j(\gp) < \max(\thru^j,\thrm^{j,k})$.
	It is then clear that $? \in \fire_j(\bias(\thru^j+\e'))$.
	It is also clear that $b \in \fire_j(\bias(\thru^j+\e'))$. 
	Notice that by Lemma~\ref{lem:any-qbias-from-bias} we can
	get $\biasq(\th)$ (for any $\th$) as the tuple before round 
	$\fo+1$. Choose $\th = \thru^{\fo+1} + \e$ for some small
	$\ep$.
	It is clear that we can construct a multi-set $\ho$ of size $1-\ep$
	consisting of $\thru^{\fo+1}$ fraction of $b$'s and 
	the remaining as $?$'s from the tuple $\bias(\th)$. Notice that $\ho$
	does not satisfy any instructions and (for a small enough $\e$)
	is bigger than $\thr_{\fo+1}(\gp)$. 
	Further by sending the entire tuple as a multi-set we get that
	$b \in \fire_{\fo+1}(\biasq(\th),\gp)$. Hence
	$\set{b,?} \subseteq \fire_{\fo+1}(\biasq(\th),\gp)$.\\
	
	
	
	
	In both cases, we can then use Lemma~\ref{lem:any-qbias-from-bias}
	to ensure that half the processes remain undecided and half the processes decide on $b$. 
	Further, in both cases, we can arrange the execution in such a way that the state after this phase is either $(\bias(\th),\spread^?)$ or $(\spread,\spread^?)$. 
	
	If the state is $(\spread,\spread^?)$ then by Lemma~\ref{lem:spread} $a \in \fire_1(\spread,\gp)$ 
	and so in the next phase we can get $\soloa$ as the tuple after the first
	round and make the other processes decide on $a$. 
	In the remaining case we consider separately the two conditions on constants
	that can be violated.
	
	If $\thrmok/2<1-\thru^{\fo+1}$ then in the first round of the next phase
	consider the $\ho$ set containing all the $a$'s in $\bias(\th)$ and the number of $b$'s
	smaller by $\e$ than the number of $a$'s.
	The size of this set is $(1-\th)+(1-\th-\e)=2(1-\thru^{\fo+1})-3\e$. 
	For a suitably small $\e$, this quantity is bigger than $\thrmok$
	which by observation~(\ref{eq:syntactic-property}) is bigger
	than $\thr_1(\p)$. 
	So we can get $\solo^a$ as the tuple after the first round and then use this to make the undecided processes decide on $a$.

	If $\thruo<1-\thru^{\fo+1}$, then just take $\ho$ set in $\bias(\th)$ 
	consisting only of all the $a$'s. 
	Once again for a small enough $\e$, the size of this set is bigger than $\thruo$ which by
	observation~(\ref{eq:syntactic-property}) is bigger than $\thr_1(\p)$.
	Hence, we can get $\solo^a$ as the tuple after the first round 
	and use this to make the undecided processes decide on $a$.
\end{proof}

\begin{lemma}\label{lem:agreement}
	If all the structural properties are satisfied then the algorithm satisfies agreement.
\end{lemma}

\begin{proof}
	It is clear that the algorithm satisfies agreement when the state of 
	the $\inp$ variable
	is either $\solo$ or $\solo^a$.  Suppose 
	we have an execution $(\bias(\theta),d) \act{\gp}\cdots\act{\gp}
	(\bias(\theta'),d')$ 
	such that $(\bias(\theta'),d')$ is the first state in this execution
	with a process $p$ which has decided on a value. 
	We consider the case when $a$ is this value.
	The other case is analogous.
	
	Since
	$\thrm^{1,k}/2 \geq 1-\thru^{\fo+1}$ we have that $\thru^{\fo+1} \ge 1/2$. Further round $\fo+1$ does not have a $\mult$ instruction.
	It then follows directly from the semantics that if $q$ is a process
	then either $d'(q) = a$ or $d'(q) = ?$.
	Further notice that since $a$ was decided by some process, it has to be 
	the case that at least $\thru^{\fo+1}$ processes have $a$ as their $\inp$ value. Hence $\theta' < 1 - \thru^{\fo+1}$.
	
	Since $\theta' < 1 - \thru^{\fo+1} \le \thru^1$, it follows that $b$ cannot be fired from $\bias(\theta')$ using the $\uni$ instruction in the first round.
	Since $\theta' < 1 - \thru^{\fo+1} \le \thrm^{1,k}/2$ and since every $\mult$ instruction in the first round has $\smor$ as its operator, it follows that
	$b$ cannot be fired from $\bias(\theta')$ using the $\mult$ instruction as well. 
	Hence the number of $b$'s in the $\inp$ tuple can only decrease from this point onwards and so it follows that no process from this point onwards can decide on $b$.
	The same argument applies if there are more than two values.
\end{proof}

\subsubsection*{Part 2: termination}
We consider only two values $a,b$. 
It is direct form the arguments below that the termination proof also works if
there are more values. 

\begin{lemma}\label{lem:fire-one}
	For the global predicate $\gp$: $a\in\fire_1(\bias(\th),\gp)$ iff $\th<\bthr$. 
	(Similarly $b\in\fire_1(\bias(\th),\gp)$ iff $1-\bthr<\th$).
\end{lemma}
\begin{proof}
	In order for a multi-set $\ho$ to be such that $a = \update_1(\ho)$
	there are two possibilities: (i) it should be of size $>\thruo$ and contain only $a$'s, or (ii) of size
	$>\thrmok$ and contain at least as many  $a$'s as $b$'s. 
	Recall that by observation~(\ref{eq:syntactic-property}) on page~\pageref{eq:syntactic-property} we have
	$\thru^1 \ge \thr_1(\gp)$ and $\thrm^{1,k} \ge \thr_1(\gp)$.

	
	The number of $a$'s in $\bias(\th)$ is $1-\th$.
	So the first case is possible only iff $1-\th>\thruo$, i.e., when $\th<1-\thruo$. 
	Further if $\th < 1-\thruo$, then we can send a set $\ho$ consisting only of
	$a$'s, such that  $|\ho| > \thru^1 \ge \thr_1(\gp)$ 
	and so $a \in \fire_1(\bias(\th),\gp)$.
	The second case,  is possible only if $1-\th>\thrmok/2$, or equivalently,  $\th<1-\thrmok/2$. Further if $\th < 1-\thrmok/2$, 
	then we can send a set $\ho$ of size $\thrmok$ consisting of 
	both $a$'s and $b$'s, which will ensure that $a \in \fire_1(\bias(\th),\gp)$.
	To sum up, $a \in \fire_1(\bias(\th),\gp)$ iff $\th<\bthr$.
	The proof for $b \in \fire_1(\bias(\th),\gp)$ iff $1-\bthr < \th$ is similar.
\end{proof}

\begin{corollary}\label{cor:no-a-above-bthr}
	For every predicate $\p$, if $\th\geq\bthr$ then $a\not\in\fire_1(\bias(\th),\p)$.
	Similarly if $\th\leq 1-\bthr$ then $b \not \in\fire_1(\bias(\th),\p)$.
\end{corollary}
\begin{proof}
	We have assumed that every predicate implies the global predicate, so every
	$\ho$ set that is admissible w.r.t.\ some predicate, is also admissible
	w.r.t.\ the 	global predicate. 
	Lemma~\ref{lem:fire-one}, says that $a$ cannot be obtained under the global predicate if $\th \geq \bthr$.
	Similar proof holds for the other claim as well.
\end{proof}


\begin{lemma}\label{lem:unifier}
	Suppose $\p$ is a unifier and  $\bias(\th)\act{\p}f$.
	If $\thruo\leq \thrmok$ or $1-\bthr\le\th\le\bthr$ 
	then $f=\solo$ or $f=\solo^a$.
\end{lemma}
\begin{proof}
	We first show that the value $?$ cannot be produced in the first round. 
	Since $\p$ is a unifier we have $\thr_1(\p) \geq \thrmok$.
	If $\thr_1(\p) \geq \thruo$ then we are done. 
	Otherwise $\thr_1(\p)< \thruo$, implying $\thr_1(\p)\geq \bthr$, by the definition
	of unifier. 
	We consider $1-\bthr\le\th\le\bthr$, and the tuple $\bias(\th)$.
	In this case, every heard-of multiset $\ho$ strictly bigger than the
	threshold $\bthr$ (and hence bigger than $\thr_1(\p)$) 
	must contain both $a$ and $b$. 
	Since there is a $\mult$ instruction in the first round
	(and since $\thr_1(\p) \ge \thrmok$), the first
	round cannot produce $?$, i.e., after the first round 
	the value of the variable $x_1$ of each process is either $a$ or $b$.
	
	Let $i$ be the round such that $\p_i$ is an equalizer
	and rounds $2,\dots,i$ are non-preserving. 
	This round exists by the definition of a unifier.
	Thanks to above, we know that after the first round no process 
	has $?$ as their $x_1$ value. Since rounds $2,\dots,i$ are non-preserving, it follows that 
	till round $i$ we cannot produce $?$ under the predicate $\p$.
	Because $\p_i$ has an equalizer, after round $i$ we either have the tuple $\solo$ or $\solo^a$.
	This tuple stays till round $\fo$ as the rounds  $i+1,\dots,\fo$ are solo-safe. 
\end{proof}

Observe that if rounds $\fo+1,\dots,r$ of a unifier $\p$ are solo-safe then $\p$ is also a
decider and all processes decide.
Otherwise some processes may not decide.
So unifier by itself is not sufficient to guarantee termination.

\begin{lemma}\label{lem:decider}
	If $\p$ is a decider and $(\solo,\solo^?)\act{\p}(f',d')$ 
	then $(f',d') = (\solo,\solo)$.
	Similarly, if $(\solo^a,\solo^?) \act{\p}(f',d')$ then $(f',d') =
	(\solo^a,\solo^a)$.
	In case $\thrmok\leq \thruo$, for every $\th\geq\bthr$: if
	$(\bias(\th),\solo^?)\act{\p}(f',d')$ then $(f',d')=(\solo,\solo)$
	and for every $\th\leq 1-\bthr$: if
	$(\bias(\th),\solo^?)\act{\p}(f',d')$ then $(f',d')=(\solo^a,\solo^a)$.
\end{lemma}
\begin{proof}
	The first two statements are direct from the definition as all the rounds in a
	decider are solo safe. 	
	We only prove the third statement, as the proof of the fourth statement is similar.
	For the third statement, by Corollary~\ref{cor:no-a-above-bthr} after the first round we
	cannot produce $a$'s under $\p_1$.
	Because  the first round is solo-safe, we get
	$\thruo\leq\thr_1(\p)$; and since $\thrmok\leq \thruo$, we get $\thrmok\leq \thr_1(\p)$. 
	Hence, the first round cannot produce $?$ neither.
	This means that from $\bias(\th)$ as the input tuple, we can 
	only get $\solo$ as the tuple after the first round under the
	predicate $\p_1$.
	Since rounds $2,\dots,r$ are solo-safe it follows that 
	all the processes decide on $b$ in round $r$.
\end{proof}





We are now ready to show one direction of Theorem~\ref{thm:core}.

\begin{lemma}\label{main-positive}
	If an algorithm in a core language has structural properties from
	Definition~\ref{def:structure}, and satisfies condition T then it solves consensus.
\end{lemma}
\begin{proof}
	Lemma~\ref{lem:agreement} says that the algorithm satisfies agreement.
	If condition T holds, there is a  unifier followed by a decider.
	If $\thruo\leq \thrmok$ then after a unifier the $\inp$ tuple becomes
	$\solo$ or $\soloa$ thanks to Lemma~\ref{lem:unifier}.  
	After a decider all processes decide thanks to Lemma~\ref{lem:decider}.
	
	Otherwise $\thrmok< \thruo$. 
	If before the unifier the $\inp$ tuple was $\bias(\th)$ with $1-\bthr\le\th\le\bthr$ then after the
	unifier $\inp$ becomes $\solo$ or $\soloa$ thanks to Lemma~\ref{lem:unifier}. 
	We once again conclude as above. 
	If $\th>\bthr$ (or $\th < 1-\bthr$) then by Corollary~\ref{cor:no-a-above-bthr}, the number of $b$'s (resp. number of $a$'s) can only
	increase after this point. Hence till the decider, 
	the state of the $\inp$ tuple remains as $\bias(\th')$ with 
	$\th' > \bthr$ (resp. $\th' < 1-\bthr$).
	After a decider all processes decide thanks to Lemma~\ref{lem:decider}.
\end{proof}

\subsubsection*{Part 3: non-termination}

\begin{lemma}\label{lem:not-decider}
	If $\p$ is a not a decider then	
	$\solo\act{\p}\solo$ and $\solo^a \act{\p} \solo^a$; namely, no process may decide.
\end{lemma}
\begin{proof}
	If $\p$ is not a decider then there is a round, say $i$, that is not solo-safe. 
	By definition this means $\thr_i(\p)< \thru^i$.
	It is then easy to verify that for $j < i$, $\solo \lact{\p_j}_j \solo$,
	$\solo \lact{\p_i}_i \solo^?$ and $\solo^? \lact{\p_k}_k \solo^?$ for $k > i$. Hence this
	ensures that no process decides during this phase.
	Similar proof holds when the $\inp$ tuple is $\solo^a$.
\end{proof}

\begin{lemma}\label{lem:global-th}
	For the global predicate $\gp$: if $1/2\leq\th<\bthr$, then 
	$\bias(\th)\act{\gp}\bias(\th')$ for every $\th'\geq 1/2$.
\end{lemma}
\begin{proof}
	We first observe that $a,b \in \fire_1(\bias(\th),\gp)$.
	Indeed, by Lemma~\ref{lem:fire-one}, $a\in \fire_1(\bias(\th),\gp)$. 
	Further since $1/2 \le \th$ and every $\mult$ instruction
	in the first round has $\smor$ as operator, it 
	follows that $b \in \fire_1(\bias(\th),\gp)$.
	Recall that by our proviso, the global predicate is not an equalizer.
	
	
	Suppose there are $\mult$ instructions in rounds $2,\dots,\fo$.
	Then Lemma~\ref{lem:any-frequency} allows us to get 
	$\bias(\th')$  as the tuple after round $\fo$.
	Moreover, our proviso from page~\pageref{proviso} says that there is no
	$\mult$ instruction in round $\fo+1$.
	So we can get $\soloq$ as the tuple after round $\fo+1$ by sending the whole multiset. We can then propagate the tuple $\soloq$ all 
	the way till the last round.
	This ensures that no process decides and we are done in this case. 
	
	Otherwise there is a round $j$ such that $2 \le j \le \fo$ and $j$ does not have any $\mult$ instruction.
	By Lemma~\ref{lem:any-qbias-from-bias} 
	we can get $\biasq(\th'')$ as well as $\biasqa(\th'')$ (for any $\th''$) after
	round $\fo$.  
	There are two cases depending if $\th'\geq\th$ or not. 
	
	If $\th'\geq \th$, then we consider the tuple $\biasq(\th'')$ for some $1/2\le \th''<
	\min(\thru^{\fo+1}-\e,\th)$, (where $\e$ is some small number). 
	Notice that by Lemma~\ref{lem:constants} we have $\thrmok/2 \geq 1-\thru^{\fo+1}$ and since
	$\thrmok < 1$, this implies that $\thru^{\fo+1} > 1/2$, and so such a $\th''$ exists.
	It is  clear that $? \in \fire_{\fo+1}(\biasq(\th''),\gp)$
	and so we can get $\soloq$ as the tuple after round $\fo+1$ thereby
	ensuring that no process decides.
	To terminate, we need to arrange this execution so that the state of $\inp$
	becomes $\bias(\th')$ after this phase.
	Since $\th''\ge1/2$ we have enough $b$'s to change $\th'-\th$ fraction of $a$'s
	to $b$'s. 
	We leave the other values unchanged. 
	This changes the state of $\inp$ from $\bias(\th)$
	to $\bias(\th')$.
	
	
	Suppose $\th' < \th$.
	By Lemma~\ref{lem:any-qbias-from-bias}, after round $\fo$ we can reach the
	tuple 
	$\biasqa(\th'')$ for $\th''=\th-\th'$.
	Arguing as before, we can ensure that the 
	state of the $\inp$ can be converted to $\bias(\th')$. We just have
	to show that all processes can choose to not decide in the last
	round. 
	
	We observe that $\th'' \le \thru^{\fo+1}$. 
	Indeed since $\th < \bthr \le 1$ and $\th' \geq 1/2$, it
	follows that $\th'' < 1/2 \le \thru^{\fo+1}$, where the last inequality
	follows from the discussion in the previous paragraph. 
	Now, as $\th'' \le \thru^{\fo+1}$, if we send 
	the entire tuple $\bias_a^?(\th'')$ to every process, we get $\soloq$ as the tuple after
	round $\fo+1$, hence 	making the processes not decide on anything in the last
	round.
	
\end{proof}




\begin{lemma}\label{lem:not-uni}
	If $\p$ is not a  unifier then 
	\begin{equation*}
	\bias(\th)\act{\p}\bias(\th) \quad\text{for some $1/2\leq\th<\bthr$.}
	\end{equation*}
\end{lemma}
\begin{proof}
	We examine all the reasons why $\p$ may not be a unifier.
	
	First let us look at conditions on constants. 	
	If $\thr_1(\p)<\thrmok$ then let $\th=1/2$. In the first round, we can then send to every process a multi-set
	$\ho$ with both $a$'s and $b$'s, and of size in between $\thr_1(\p)$ and $\thrmok$.
	This allows us to get $\soloq$ as the tuple after the first round, and
	ensures that neither the $\inp$ tuple nor the $\dec$ tuple gets
	updated in this phase.
	
	Suppose $\thr_1(\p)<\thruo$ and $\thr_1(\p) < \bthr$. 
	Let $\e$ be such that $\thr_1(\p) + \e < \min(\thruo,\bthr)$
	and let $\th = \max(\thr_1(\p)+\e,1/2)$.
	In the first round, by sending to every process a fraction of $(\thr_1(\p)+\e)$ $b$'s from $\bias(\th)$ we get $\soloq$ as the tuple after the first round
	and that allows us to conclude as before.

	The second reason is that there is no equalizer in $\p$ up to round $\fo$.
	We take $\th=1/2$. 
	By Lemmas~\ref{lem:spread}	and~\ref{lem:no-min}, we have
	$a,b\in\fire_1(\spread,\p)$.  
	If all the rounds $1,\dots,\fo$ have a $\mult$ instruction then
	Lemma~\ref{lem:any-frequency} allows us to get $\spread$ as the tuple after round $\fo$.
	Lemma~\ref{lem:no-mult-sequence} says that there cannot be a $\mult$ instruction
	in round $\fo+1$, so by sending the whole multiset in this round, we get
	$\soloq$ as the tuple after round $\fo+1$. This 
	ensures that no process decides in this phase.
	The other case is when there is a round among $1,\dots,\fo$ without the $\mult$ instruction.
	Lemma~\ref{lem:any-qbias-from-bias} allows to get $\soloq$ after round $\fo$ and 
	so neither $\inp$ nor $\dec$ of any process gets updated.
	
	The last reason is that there is  a round before an equalizer that is
	preserving, or a 	round after the equalizer that is not solo-safe. 
	In both cases we can get $\soloq$ as the tuple at round $\fo$ and conclude as before.
\end{proof}

The next lemma gives the main non-termination argument.
\begin{lemma}
	If the structural conditions from Definition~\ref{def:structure} hold, but condition T
	does not hold then the algorithm does not terminate. 
\end{lemma}
\begin{proof}
	We recall that the communication predicate is:
	\begin{equation*}
		(\lG\gp)\land(\lF(\p^1\land\lF(\p^2\land\dots (\lF\p^k)\dots)))
	\end{equation*}	
	and that we have assumed that the global predicate implies all sporadic predicates. 
	This means, for example, that if the global predicate is a decider then all
	sporadic predicates are deciders. 
	
	We construct an execution
	\begin{equation*}
	f_1\act{\p^1}f'_1\act{\gp}f_2\act{\p^2}f'_2\act{\gp}\dots \act{\p^k}f'_k\act{\gp} f'_k
	\end{equation*}
	where every second arrow is a transition on  the global predicate. 
	The last	transition on the global predicate is a self-loop.
	Recall that we write $f\act{\p}f'$ for $(f,\soloq)\act{\p}(f',\soloq)$, so
	indeed the run as above is a witness to non-termination. 
	
	We examine several cases.
	
	If none of $\p^1,\dots,\p^k,\gp$ is a decider, then we take $f_i=f'_i=\solo$ for
	all $i=1,\dots,k$.
	By Lemma~\ref{lem:not-decider} we get the desired execution.
	
	Suppose the last decider in the sequence $\p^1,\dots,\p^k$ is $\p^l$. 
	(Notice that if the global predicate $\p$ is a decider then $l=k$). 
	By our assumption, none of $\p^1,\dots,\p^l$ are unifiers. 
	By Lemma~\ref{lem:not-uni},  for every $\p^i$, $i=1,\dots,l$, there is
	$1/2\leq\th_i<\bthr$ such 	that $\bias(\th_i)\act{\p^i}\bias(\th_i)$. 
	So we take $f_i=f'_i=\bias(\th_i)$.
	We can then use Lemma~\ref{lem:global-th} to get $f'_i \act{\gp} f_{i+1}$, for
	all $i=1,\dots,l-1$.
	This gives us an execution up to $f'_l$. 
	
	To complete the execution we consider two cases. 
	If $l = k$, then by Lemma~\ref{lem:global-th} we have 
	$f'_k\act{\gp}f'_k$ and so we are done.
	Otherwise $l<k$, and we use Lemma~\ref{lem:global-th} to get
	$f'_l\act{\gp}\solo$. 
	We set $f_j=f'_j=\solo$ for $j>l$.
	Since $l < k$, neither the global predicate, nor any one of $\p^{l+1},\dots,\p^k$ are
	deciders, and so by Lemma~\ref{lem:not-decider} we get
	the desired execution.
\end{proof}
 
\section{Proofs for algorithms with timestamps}

We prove the characterization from Theorem~\ref{thm:ts}.
Recall that in this extension we add timestamps to the $\inp$ variable, i.e., timestamps are sent along with $\inp$ and are updated whenever $\inp$ is updated. 
The semantics of rounds is different only in the first round where we have 
$(f_0,t) \lact{}_1 f_1$ instead of the  $f_0 \lact{}_1 f_1$ in the core language.
Further, whenever the $\inp$ of a process is updated, the timestamp is updated as well. 
(In particular, if the value of $\inp$ of a process was  $a$ and later it was updated
to $a$ once again, then in principle the value of $\inp$ does not change but the
time stamp is updated).

\begin{definition}
	We introduce some abbreviations for tuples of values with timestamps. 
	For every tuple of values $f$, and every  $i \in \nn$ define
	$(f,i)$ to be a $\inp$-timestamp tuple where the value of $\inp$ for process
	$p$ is $f(p)$ and the value of the timestamp is $i$. So, for example,
	$(\spread,0)$ denotes the tuple where the value of $\inp$ for half of the
	process is $a$, for the other half it is $b$, and the
	timestamp for every process is $0$.
\end{definition}

Similarly to the core case, we give the proof of Theorem~\ref{thm:ts} in three
parts: we deal first with structural properties and then with termination, and non-termination.

\subsubsection*{Part 1: Structural properties for timestamps}
The structure of the argument is similar to the core case.

\begin{lemma}\label{lem:ts-no-mult}
	If no $\mult$ instruction in the first round then the algorithm does not have
	termination property.
\end{lemma}
\begin{proof}
	It is easy to verify that $(\spread,0) \act{\p} (\spread,0)$ is a phase transition, for
	every communication predicate $\p$. 
\end{proof}

\begin{lemma}\label{lem:ts-no-uni}
	If there is a round without uni instruction then the algorithm does not have
	termination property.
\end{lemma}
\begin{proof}
	Let $l$ be the first round without uni instruction. If $l \le \fo$,
	we have $(\solo_a,0)\act{\p}(\solo_a,0)$ for every communication predicate.
	Otherwise we get $(\solo_a,i) \act{\p} (\solo_a,i+1)$ for every communication predicate.
\end{proof}

The next Lemma points out a crucial difference with the case without timestamps (cf.~Lemma~\ref{lem:fire-one})

\begin{lemma}\label{lem:ts-bias-fire}
	For every
	$\p$, we have $\set{a,b}\incl\fire_1((\bias(\th),i),\p)$ for every $i$
	and sufficiently big $\th$.
\end{lemma}

\begin{proof}
	We let $\th = \max(\thruo,\thr_1(\p)) + \e$ for small enough $\e$. 
	So $b \in \fire_1((\bias(\th),i),\p)$, when we take $\ho$ to contain all
	the $b$'s.
	Since all the $\mult$ instructions have $\maxts$ as their operator,
	it follows that if we take a multi-set $\ho$ consisting of all the values in
	the tuple then $a = \update_1(\ho)$.
\end{proof}

Since the semantics of the rounds remains the same except for the first one,
Lemmas~\ref{lem:any-qbias-from-qbias},~\ref{lem:any-frequency} 
and~\ref{lem:any-qbias-from-bias} apply for timestamp algorithms for $k>2$. 
For the first round, we get the following reformulations.

\begin{lemma}\label{lem:ts-any-frequency}
	Suppose rounds $1\dots l$ all have $\mult$
	instructions
	and none of $\p_1,\dots,\p_l$ is an equalizer. 
	If $\dom(f')\incl \fire_1((f,t),\p_1)$ and  $? \notin \dom(f')$ 
	then 
	$(f,t)\lact{\p_1}_1\dots\lact{\p_l}_l f'$ is possible.
\end{lemma}
\begin{proof}
	Same as that of Lemma~\ref{lem:any-frequency}.
\end{proof}

\begin{lemma}\label{lem:ts-any-bias}
	Suppose none of $\p_1,\dots,\p_l$ is an equalizer, and some round $1,\dots,l$
	does not have a $\mult$ instruction. 
	For every $\th$ and every $(f,t)$ such that $\set{a,b}\in
	\fire_1((f,t),\p_1)$ or $\set{b,?}\in\fire_1((f,t),\p_1)$	we have
	$(f,t)\lact{\p_1}_1\dots\lact{\p_l}_l \biasq(\th)$.  
\end{lemma}
\begin{proof}
	The same argument as for Lemmas~\ref{lem:any-qbias-from-qbias}
	and~\ref{lem:any-qbias-from-bias}; replacing replace Lemma~\ref{lem:any-frequency} with Lemma~\ref{lem:ts-any-frequency}.
\end{proof}

We can now deal with the case when there is $\mult$ instruction in round $\fo$. 
This is an analog of Lemma~\ref{lem:no-mult-sequence}.

\begin{lemma} \label{lem:ts-min}
	Suppose round $\fo$ either contains a $\mult$ instruction or $\thru^{\fo} < 1/2$. 
	Then either the algorithm violates consensus or we can remove the 
	$\mult$ instruction and make $\thru^{\fo} = 1/2$ without 
	affecting the semantics of the algorithm.
\end{lemma}

\begin{proof}
	Suppose round $\fo$ either contains a $\mult$ instruction or
	$\thru^{\fo} < 1/2$. We consider two cases:
	
	The first case is when there does not exist any tuple $(f,t)$ with
	$(f,t) \lact{\gp_1}_1 f_1 \dots \lact{\gp_{\fo-2}}_{\fo-2} f_{\fo-2}$
	such that $a,b \in \fire_{\fo-1}(f_{\fo-2},\gp)$. It is then clear
	that the $\mult$ instructions in round $\fo$ will never be fired
	and so we can remove all these instructions in round $\fo$.
	Further it is also clear that setting $\thru^{\fo} = 1/2$ 
	does not affect the semantics of the algorithm in this case.
	
	So it remains to examine the case when  there exists a tuple $(f,t)$ with
	$(f,t) \lact{\gp_1}_1  f_1\lact{\gp_2} \cdots\lact{\gp_{\fo-2}} f_{\fo-2}$ such that
	$a,b \in \fire_{\fo-1}(f_{\fo-2},\gp)$.
	It is clear that in this case, we also have,
	$(\bias(\thru^1+\e),0) \lact{\gp_1}_1  f_1\lact{\gp_2} \cdots\lact{\gp_{\fo-2}} f_{\fo-2}$. Also we get
	$f_{\fo-2} \lact{\gp_{\fo-1}}_{\fo-1} \bias(\th)$ for arbitrary $\th$.
	
	In this case, we will show the following:
	Depending on the structure of rounds $\fo$ and $\fo+1$ we will define 
	two tuples $f_{\fo-1}$ and $f_{\fo}$ with the following properties:
	\begin{itemize}
		\item $f_{\fo-2} \lact{\gp_{\fo-1}}_{\fo-1} f_{\fo-1} \lact{\gp_{\fo}}_{\fo} f_{\fo}$,
		\item $f_{\fo}$ contains no $?$ and at least one $a$,
		\item either $a,b \in \fire_{\fo+1}(f_{\fo},\gp_{\fo+1})$
		or $b,? \in \fire_{\fo+1}(f_{\fo},\gp_{\fo+1})$.
	\end{itemize} 
	Notice that if $a,b \in \fire_{\fo+1}(f_{\fo},\gp_{\fo+1})$ and all rounds after 
	round $\fo$ have a $\mult$ instruction, then we can use Lemma~\ref{lem:any-frequency}
	to conclude that we can decide on both $a$ and $b$.
	In the other case, i.e., if some round after round $\fo$ does not have a $\mult$
	instruction, or $b,? \in \fire_{\fo+1}(f_{\fo},\gp_{\fo+1})$
	we use Lemmas~\ref{lem:any-qbias-from-bias} and~\ref{lem:any-qbias-from-qbias}
	to show that we can make half the processes decide on $b$
	and the other half undecided. Now the state of
	the algorithm after this phase will be $(f_{\fo},1,\spread^?)$
	where $f_{\fo}$ contains at least one $a$.
	Since all the $\mult$ instructions have $\maxts$ as operator
	it follows that we can then get $\solo^a$ after the first 
	round and decide on $a$.
	
	So it remains to come up with $f_{\fo-1}$ and $f_{\fo}$ with the required properties.
	We will do a case analysis, and for each case provide both these tuples.
	In each of these cases, it can be easily verified that the provided tuples
	satisfy the required properties.
	
	\begin{itemize}
		\item $\thru^{\fo} < 1/2$ or the $\mult$ instruction with the highest threshold
		in round $\fo$ has $\smor$ as operator.
		\begin{itemize}
			\item The $\mult$ instruction with the highest threshold in round $\fo+1$
			has $\smor$ as operator: Take $f_{\fo-1} = f_{\fo} = \spread$.
			\item Otherwise: Take $f_{\fo-1} = \spread, f_{\fo} = \bias(\max(\thru^{\fo+1},\thr_{\fo+1}(\gp))+\e)$.
		\end{itemize}
		\item The $\mult$ instruction with the highest threshold in round $\fo$ has $\min$
		as operator.
		\begin{itemize}
			\item The $\mult$ instruction with the highest threshold in round $\fo+1$
			has $\smor$ as operator: Take $f_{\fo-1} = \bias(\max(\thru^{\fo},\thr_{\fo}(\gp))+\e), f_{\fo} = \spread$.
			\item Otherwise: Take $f_{\fo-1} = \bias(\max(\thru^{\fo},\thr_{\fo}(\gp))+\e)$,\\
			$f_{\fo} = 
			\bias(\max(\thru^{\fo+1},\thr_{\fo+1}(\gp))+\e)$.
		\end{itemize}
	\end{itemize}

\end{proof}

\begin{corollary} \label{cor:ts-min}
	If round $\fo+1$ has a $\mult$ instruction or $\thru^{\fo+1} < 1/2$,
	then the $\mult$ instruction can be removed and $\thru^{\fo+1}$ can be 
	made $1/2$ without altering the semantics of the algorithm.
\end{corollary}

\begin{proof}
	By the previous lemma, round $\fo$ does not have any $\mult$
	instruction and $\thru^{\fo} \ge 1/2$. 
	It then follows that if $f \lact{\f}_{\fo} f'$ for arbitrary predicate $\f$
	then there cannot be both $a$ and $b$ in $f$.
	Hence the $\mult$ instruction in 
	round $\fo+1$ will never be fired.
	Consequently, it can
	be removed without affecting the correctness of the algorithm.
	It is also clear that we can raise the value of $\thru^{\fo+1}$ to $1/2$
	without affecting the semantics of the algorithm. 
\end{proof}

\begin{lemma} \label{lem:ts-constants}
	If the property of constants from Definition~\ref{def:ts-structure} 
	is not satisfied, then agreement is violated.
\end{lemma}

\begin{proof}
	The proof starts similarly to the one of Lemma~\ref{lem:constants}.
	We consider an execution under the global predicate $\gp$, and employ
	Lemmas~\ref{lem:ts-any-frequency} and~\ref{lem:ts-any-bias}. 
	We start from configuration $(\bias(\th_1),0)$ where
	$\th_1>\thruo$ big enough so that by Lemma~\ref{lem:ts-bias-fire} we get
	$\set{a,b} \incl\fire_1(\bias(\th_1),\gp)$.
	We consider also $\th=\thru^{\fo+1}+\e$.
	
	By Lemma~\ref{lem:ts-min} there is a preserving round before 
	round $\fo+1$. 
	We proceed differently depending on wether
	$\thrmok<1-\thru^{\fo+1}$ or not.
	By the same argument as in Lemma~\ref{lem:constants}, we can 
	can get $\biasq(\th)$ or $\biasqa(\th)$ after round $\fo$.
	If $\thrmok<1-\thru^{\fo+1}$  then we choose to get $\biasq(\th)$. 
	We use Lemma~\ref{lem:any-qbias-from-bias} to make some processes decide on
	$b$.
	After this phase there are $1-\th$ processes with timestamp $0$.
	We can ensure that among them there is at least one with value $a$ and one
	with value $b$. 
	Since there is $\mult$ instruction in the first round, in the next phase we
	send all the  values 	with timestamp $0$. 
	This way we get $\soloa$ after the first round, and make some process decide
	on $a$.
	
	The remaining case is when $\thrmok\geq 1-\thru^{\fo+1}$. 
	So we have $\thruo<1-\thru^{\fo+1}$, since we have assumed that that the
	property of constants from 	Definition~\ref{def:ts-structure} does not hold. 
	This time we choose to get $\biasqa(\th)$ after round $\fo$, and make some process
	decide on $a$. 
	Since we have started with $\bias(\th_1)$ we can arrange updates so that we
	have at least $\min(\th_1,1-\th)$ processes who have value $b$ with timestamp
	$0$. 
	But $\min(\th_1,1-\th)>\thruo$, so by sending $\ho$ set consisting of these
	$b$'s we reach $\solo$ after the first round and make some processes decide on
	$b$. 
\end{proof}

Now we can state the sufficiency proof, similar to Lemma~\ref{lem:agreement}.


\begin{lemma}\label{lem:ts-structure}
	If all the structural properties from Definition~\ref{def:ts-structure} are
	satisfied then the  algorithm satisfies agreement.
\end{lemma}
\begin{proof}
	It is clear that the algorithm satisfies agreement when the initial frequency
	is either $\solo$ or $\solo^a$.  Suppose $(\bias(\theta),t,d) \act{\p^*}
	(\bias(\theta'),t',d')$ such that $(\bias(\theta'),t',d')$ is the first state in this execution
	with a process $p$ which has decided on a value.
	Without loss of generality let $b$ be the value that $p$ has decided on.
  %
  %
	By Lemma~\ref{lem:ts-min} round $\fo$ does not have any 
	$\mult$ instructions and $\thru^{\fo} \ge 1/2$  and so it follows 
	that every other process could only decide
	on $b$ or not decide at all. For the same reason it follows
	that every process either updated its $\inp$ value to $b$
	or did not update its $\inp$ value at all.
	Further notice that since $b$ was decided by some
	process, it has to be the case that more than $\thru^{\fo+1}$ processes have $b$ as
	their $\inp$ value with the most recent phase as their timestamps.
	This means that the number of $a$'s in the configuration is  less than $1-\thru^{\fo+1}$. Moreover since every process either updated its
	$\inp$ to $b$ or did not update all, no process with $a$ has the latest
	timestamp. 
	
	Since $\thruo \geq 1 - \thru^{\fo+1}$, it follows that $a$ cannot be fired
	from $(\bias(\theta'),t')$ using the $\uni$ instruction in the first round. Further
	since $\thrmok\geq  1 - \thru^{\fo+1}$ it follows that any $\ho$ set bigger
	than $\thrmok$ has to contain a value with the latest timestamp. 
	Since the only value with the latest timestamp is the value $b$,
	it follows that
	$a$ cannot be fired from $(\bias(\theta'),t')$ using the $\mult$ instruction as
	well. In consequence, the number of $a$'s can only decrease from this point onwards and so it
	follows that no process from this point onwards can decide on $a$. 
\end{proof}

The proof of termination is simpler compared to the proof of termination for
core language. This is in part due to the use of $maxts$ rather than $smor$ as the
operator in the first round. 

\subsubsection*{Part 2: termination for timestamps}

\begin{lemma}\label{lem:ts-dec}
	If $\p$ is a decider and $(\solo,t,\soloq)\act{\p}(f',t',d')$
	then $(f',d') = (\solo,\solo)$ for every timestamp tuple $t$. Similarly if
	$(\solo^a,t,\soloq) \act{\p}(f',t',d')$ then $(f',d') = (\solo^a,\solo^a)$.
\end{lemma}
\begin{proof}
	Immediate
\end{proof}

\begin{lemma}\label{lem:strong-unifier}
	Suppose $\p$ is a strong unifier. If  $(\bias(\th),t)\act{\p}(f,t')$
	then $f=\solo$ or $f=\solo^a$ (for every timestamp tuple $t$).
\end{lemma}
\begin{proof}
	We first observe that the value $?$ cannot be produced in the first round. 
	Since $\p$ is a strong unifier we have $\thrmok\leq \thr_1(\p)$ 
	and $\thruo\leq \thr_1(\p)$, so every $\ho$ set above the threshold will
	satisfy an instruction of the first round. 
		
	Let $i$ be the round such that $\p_i$ is an equalizer
	and rounds $2,\dots,i$ are non-preserving. 
	This round exists by the definition of a unifier.
	Thanks to above, we know that till round $i$ we cannot produce $?$ under the predicate $\p$.
	Because $\p_i$ has an equalizer, after round $i$ we either have 	$\solo$ or $\solo^a$.
	This tuple stays till round $\fo$ as the rounds  $i+1,\dots,\fo$ are solo-safe.
\end{proof}

\begin{proof}\textbf{Main positive}
	Suppose there is a unifier followed by a decider. 
	After the strong unifier we have $\solo$ or $\soloa$ thanks to
	Lemma~\ref{lem:strong-unifier}.  
	After decider all processes decide thanks to Lemma~\ref{lem:ts-dec}.
	
\end{proof}

\subsubsection*{Part 3: Non-termination for timestamps}

\begin{lemma}\label{lem:ts-not-decider}
	If $\p$ is a not a decider then $(\solo,t)\act{\p}(\solo,t)$ and $(\solo^a,t)\act{\p}(\solo^a,t)$ for any timestamp $t$.
\end{lemma}
\begin{proof}
	If $\p$ is not a decider then there is a round (say $i$) that is not solo-safe. 
	So from both $\solo$ and $\soloa$ 
	we can reach the tuple $\soloq$ after round $i$.
	From $\soloq$ no process can decide.
\end{proof}

\begin{lemma} \label{lem:ts-not-str-uni}
	If $\p$ is not a strong unifier then $(\bias(\th),i) \act{\p} (\bias(\th),j)$
	is possible (for large enough $\th$, arbitrary $i$, and some $j$).
\end{lemma}

\begin{proof}
	Let $\th > \max(\thr_u^1,\thr_1(\p))+\e$ and so $a,b \in \fire_1(\bias(\th),i)$ by Lemma~\ref{lem:ts-bias-fire}. Suppose $\p$ is not a strong unifier. We do a
	case analysis. 
	
	Suppose $\thr_1(\p) < \thrm^{1,k}$ or $\thr_1(\p) < \thru^1$. 
	Clearly we can get $\solo^?$ as the tuple after the first round and then 
	use this to not decide on anything and retain the input tuple.
	
	Suppose $\p$ does not have  an equalizer. We can then 
	apply Lemmas~\ref{lem:ts-any-bias} 
	and~\ref{lem:ts-min}	
	to get $\soloq$ after round $\fo$
	and so we are done, because nothing is changed in the phase. 
	
	Suppose the $k$-th component of $\p$ is an equalizer and suppose
	there is a preserving round before round $k$ (it can be round 1 as well).
	Let the first preserving round before round $k$ be round $l$. 
	Since no round before round $l$ is preserving, it follows that
	all these rounds have $\mult$ instructions. 
	Hence by Lemma~\ref{lem:ts-any-frequency} we can get to $\bias(\th')$
	(where $\th' > \max(\thru^l,\thr_l(\p))$) before round $l$ 
	(Notice that if $l = 1$ then we
	need to reach $\bias(\th')$ with $\th' > \max(\thru^1,\thr_1(\p))$ which is where we start at). 
	It is clear that $\bias(\th') \lact{\p_l}_l \soloq$.
	We can then propagate
	$\soloq$ all the way down to get the phase transition $(\bias(\th),i) \act{\p}
	(\bias(\th),i)$. 
	
	Suppose the $k$-th component of $\p$ is an equalizer  and suppose
	there is a non solo-safe round $l$ after $k$. It is
	clear that we can reach $\solo$ after round $k$ and using this get $\soloq$
	after round $l$. Hence we once again get the phase transition $(\bias(\th),i)
	\act{\p} (\bias(\th),i)$.  
\end{proof}

\begin{proof}
	\textbf{Main non-termination}
	We show that if there is no strong unifier followed by a decider, then the
	algorithm may not terminate. We start with $(\bias(\th),0)$ where $\th$ 
	is large enough. If $\p$ is not a strong unifier then by 
	Lemma~\ref{lem:ts-not-str-uni} $(\bias(\th),i) \act{\p} (\bias(\th),j)$
	is possible for arbitrary $i$, and some $j$. 
	Hence if there is no strong unifier the algorithm will not terminate.
	
	Otherwise let $\p^l$ be the first strong unifier. Notice that $\p^l$ is not
	the global predicate as we have assumed the global predicate does not have
	equalizers.
	Till $\p^l$ we can maintain $\bias(\th)$ thanks to
	Lemma~\ref{lem:ts-not-str-uni}.
	Suppose $\p^l$ is not a decider. By Lemma~\ref{lem:strong-unifier} the 
	state of $\inp$
	after this phase will become $\solo$ or $\solo^a$. However, since $\p^l$ is not a
	decider, we can choose to not decide on any value. Hence we get the transition
	$(\bias(\th),i) \act{\p} (\solo,i+1)$. Now, since none of the
	$\p^{l+1},\dots,\p^k$ and neither the global predicate $\p$ are deciders,
	by Lemma~\ref{lem:ts-not-decider} we can have a transition where no decision
	happens.   
	Hence the algorithm does not terminate if there is no decider after a strong unifier.
\end{proof}

\section{Proofs for algorithms  with coordinators}

We give a proof of Theorem~\ref{thm:coordinators}. 
The structure of the proof is quite similar to the previous cases. 


\subsubsection*{Part 1:  Structural properties for coordinators}

\begin{lemma}\label{lem:c-no-uni}
	If there is a round without uni instruction then the algorithm does not terminate.
\end{lemma}
\begin{proof}
	We get $\soloa\act{}\soloa$ for every communication predicate.
\end{proof}

Compared to the core language, it is not easy to see that the first round of
an algorithm with coordinators  should have a $\mult$ instruction. 
However, this is indeed the case as we prove later.
For the moment we make an observation.

\begin{lemma} \label{lem:co-weak-mult}
	If the first round is not of type $\ls$ then the first round should have a $\mult$ instruction.
\end{lemma}

\begin{proof}
	Otherwise we have $\spread \act{\p} \spread$ for arbitrary communication
	predicate $\p$.
\end{proof}

Before considering the remaining structural requirements we state some useful
lemmas.

\begin{lemma}\label{lem:co-no-mult-fire}
	If round $k$ is not of
	type $\ls$ and does not have a $\mult$ instruction then for all sufficiently big
	$\th$ we have $\set{b,?}\in \fire_k(\bias(\th),\f)$, for arbitrary predicate
	$\f$. 
\end{lemma}
\begin{proof}
	Take $\th >\max(\thru^k,\f_k(\thr))$.
	We have $b\in\fire_k(\bias(\th_k),\f)$ because of the $\uni$ instruction. 
	We have $?\in\fire_k(\bias(\th_k),\f)$ because there is no $\mult$ instruction. 
\end{proof}

\begin{lemma}\label{lem:co-bias-fire}
	Suppose the first round has a  $\mult$ instruction with
	$\min$ as operation or is of type $\ls$. Then for the global predicate  $\gp$, we have  $\set{a,b}\incl\fire_1(\bias(\th),\gp)$ for
	sufficiently big $\th$.
\end{lemma}

\begin{proof}
	The claim is clear when the first round is of type $\ls$. Suppose the first round has a $\mult$ instruction with $\min$ as operation.
	Let $I$ be that instruction and let $\thr^I$ be the threshold
	value appearing in instruction $I$.
	
	Let $\th > \thru^1$. Notice that $b \in \fire_1(\bias(\th),\gp)$
	because of the $\uni$ instruction in the first round.
	Further notice that, from $\bias(\th)$ we can construct
	a multi-set $\ho$ having at least one $a$ and is of size just above $\thr^I$.
	Since $\gp$ is the global predicate, we know that this multi-set satisfies $\gp$
	because of assumption~\eqref{eq:syntactic-property}.
	Further it is clear that $a = \update_1(\ho)$ and so
	we have $a \in \fire_1(\bias(\th),\gp)$. 
\end{proof}

\begin{lemma}\label{lem:co-spread}
	Suppose in the first round all $\mult$ instructions have $\smor$ as
	operation.
	Then for every predicate $\p$ we have
	$\set{a,b}\incl\fire_1(\spread,\p)$.
\end{lemma}
\begin{proof}
	Same proof as Lemma~\ref{lem:spread}.
\end{proof}

\begin{lemma}\label{lem:co-bias-propagation}
	Suppose none of $\p_k,\dots,\p_l$  is a c-equalizer. 
	Suppose round $l$ is not of type $\lr$.
	Then there is $\th$ with $\biasq(\th)\lact{\p_k}_k\dots\lact{\p_l}_l\biasq(\th')$ for arbitrary $\th'$. 
\end{lemma}

\begin{proof}
	If the $k^{th}$ round is a $\ls$ round, consider arbitrary $\th$. By
	definition of transitions, we get
	$\biasq(\th)\lact{\p_k}_k\biasq(\th')$ for arbitrary $\th'$. 
	
	If the $k^{th}$ round is a $\lr$ round, take $\th=\thru^k+\e$ for small $\e$.
	We can get $b$ from $\biasq(\th)$ because of the $\uni$ instruction. 
	Since this is an $\lr$ round we have $\biasq(\th) \lact{\p_k}_k one_b$.
	This round must be followed by an $\ls$ round, so the argument from the previous
	paragraph applies, and we can get arbitrary $\biasq(\th')$ after round $k+1$. 
	
	Otherwise $k^{th}$ round is neither $\ls$ nor $\lr$. 
	By Lemma~\ref{lem:any-qbias-from-qbias}, we can get arbitrary $\biasq(\th')$
	after round $k$. 

	We can repeat this argument till round $l$.
\end{proof}

\begin{lemma}\label{lem:co-any-frequency}
	Suppose none of $\p_k,\dots,\p_l$ is a c-equalizer, and all rounds $k\dots l$
	have $\mult$ instructions. Suppose round $l$ is not of type $\lr$.
	For every $f$ and every $f'$ without $?$ such that $\dom(f')\incl
	\fire_k(f,\p_k)$ we have $f\lact{\p_k}_k\dots\lact{\p_l}_l f'$.
\end{lemma}
\begin{proof}
	Notice that since all the considered rounds have $\mult$ instructions, 
	none of these rounds are of type $\ls$ by assumption on
	page~\pageref{assumption-ls-lr}. 
	Further, since every $\lr$ round is followed by a $\ls$ round, it follows that
	we have only two  cases:
	either all rounds $k,\dots,l$ are of type $\every$, or
	rounds $k,\dots,l-1$ are of type $\every$ and round $l$ is of type $\lr$.
	Since the second case is excluded by assumption, we only have the 
	first case which holds by Lemma~\ref{lem:any-frequency}.
\end{proof}

\begin{lemma}\label{lem:co-any-bias}
	Suppose none of $\p_k,\dots,\p_l$ is an equalizer, and some round $k,\dots,l$
	does not have a $\mult$ instruction. 
	Suppose round $l$ is not of type $\lr$.
	For every $\th$ and every $f$  such that $\set{a,b}\in
	\fire_k(f,\p_k)$ we have 	$f\lact{\p_k}_k\dots\lact{\p_l}_l \biasq(\th)$, 
	and $f\lact{\p_k}_k\dots\lact{\p_l}_l \biasqa(\th)$.
\end{lemma}
\begin{proof}
	Let $i$ be the first round without a $\mult$ instruction.
	There are two cases. 
	
	Suppose rounds $k,\dots,i-1$ are all of type $\every$.
	In this case we use Lemma~\ref{lem:co-any-frequency} to reach any 
	$\bias(\theta')$ before round $i$. 
	If round $i$ is of type $\every$ then we can use 
	Lemma~\ref{lem:co-no-mult-fire} 
	to get arbitrary $\biasq(\th'')$ after round $i$.
	If round $i$ is of type $\lr$, then round $i+1$ is of type $\ls$ and
	so we can use Lemma~\ref{lem:co-no-mult-fire} to get 
	$one_b$ after round $i$ and (since $\p_{i+1}$ is not a c-equalizer)
	then use that to get arbitrary $\biasq(\th)$
	after round $i+1$.
	If round $i$ is of type $\ls$, since $\p_i$ is not a c-equalizer, 
	so we can get arbitrary $\biasq(\th)$ after round $i$.
	We can then use Lemma~\ref{lem:co-bias-propagation} to finish
	the proof.

	In the remaining case, by the same reasoning as in the previous lemma we see
	that all rounds $k,\dots,i-2$ must be of type $\every$, and round $i-1$ must be
	of type $\lr$.
	We can use Lemma~\ref{lem:co-any-frequency} to reach $\bias(\max(\thru^{i-1},\thr(\p_{i-1}))+\e)$
	before round $i-1$ and then using that reach $one_b$
	before round $i$. Since $i-1$ is of type $\lr$, $i$ is of type 
	$\ls$ and since there are no c-equalizers we can get 
	arbitrary $\biasq(\th)$ after round $i$. 
	We can then use Lemma~\ref{lem:co-bias-propagation} to finish
	the proof.

\end{proof}

\begin{lemma}\label{lem:co-no-mult-fo}
	If round $\fo+1$ contains a $\mult$ instruction  then the algorithm 
	does not satisfy agreement, or it can be removed without altering the
	semantics of the algorithm
\end{lemma}

\begin{proof}
	Suppose round $\fo+1$ contains a $\mult$ instruction.
	Recall that this implies that round $\fop$ is not of type $\ls$ (cf.\
	assumption on page~\pageref{assumption-ls-lr}).
	Recall that $\gp$ denotes the global predicate.

	The first case is when there does not exist any tuple $f$ having an execution 
	$f \lact{\gp_1}_1 f_1 \dots \lact{\gp_{\fo-1}}_{\fo-1} f_{\fo-1}$
	with $a,b \in \fire_\fo(f_{\fo-1},\gp)$. It is then clear
	that the $\mult$ instructions in round $\fo+1$ will never be fired
	and so we can remove all these instructions in round $\fo+1$.
	
	So it remains to examine the case when  there exists a tuple $f$ with
	$f \lact{\gp_1}_1  f_1\lact{\gp_2} \cdots\lact{\gp_{\fo-1}} f_{\fo-1}$ such that
	$a,b \in \fire_{\fo}(f_{\fo-1},\gp_{\fo})$. 
	Notice that	in this case, the first round cannot be of type $\ls$.
	Since round $\fo$ cannot be of type $\lr$ (cf. assumption on
	page~\pageref{assumption-ls-lr}) 	we can get
	$f_{\fo-1} \lact{\gp_{\fo}}_{\fo} \bias(\th)$ for arbitrary $\th$.

	We now consider two cases: Suppose round $\fo+1$ is of type $\every$.
	Then we can proceed exactly as the proof of Lemma~\ref{lem:no-mult-sequence}
	and show that agreement is not satisfied.
	
	Suppose round $\fo+1$ is of type $\lr$. Hence round $\fo+2$ is
	of type $\ls$. 
	Let $I$ be the $\mult$ instruction in round $\fo+1$ with the
	highest threshold value.
	Suppose $I$ has $\smor$ as its operation. Then we consider
	$f_{\fo-1} \lact{\gp_{\fo}}_{\fo} \spread$. 
	Since $I$ has $\smor$ as operation, it is easy to see that
	$\spread \lact{\gp_{\fo+1}}_{\fo+1} \one_b$.
	Since $\gp$ is the global predicate, $\gp_{\fo+2}$ is not a c-equalizer,
	and so we get $\one_b \lact{\gp_{\fo+2}}_{\fo+2} \biasq(\th')$ 
	for arbitrary $\th'$. We can
	then use Lemma~\ref{lem:co-bias-propagation} to conclude that
	we can make one process decide on $b$ and leave the rest 
	undecided.
	In the next phase, the state of $\inp$ is $\spread$. 
	We know, by Lemma~\ref{lem:co-weak-mult} that the first round has a $\mult$
	instruction (since as observed above, the first round
	is not $\ls$ in this case). 
	This instruction has $\smor$ (or) $\min$ as its operation,
	it is clear that in either case, $a \in \fire_1(\spread,\gp)$ and
	so we can get $\solo^a$ after the first round and decide on $a$.
	
	Suppose $I$ has $\min$ as its operation. Then we consider
	$f_{\fo-1} \lact{\gp_{\fo}}_{\fo} \bias(\th)$ where $\th > \thru^1$ 
	is sufficiently big.
	If we send the entire tuple as a HO set, we can fire $a$. 
	Hence we get $\bias(\th) \lact{\gp_{\fo+1}}_{\fo+1} \one_a$.
	As in the previous case this allows us to make one process decide on $a$. 
	Note that the state of $\inp$ will be $\bias(\th)$
	after the end of the phase.
	Since the first round has a $\uni$ instruction (Lemma~\ref{lem:c-no-uni}), and
	since $\th > \thru^1$ (and $\thru^1 \ge \thr_1(\gp)$ by equation~\ref{eq:c-syntactic-property}), we can get $\solo$ after the first round
	and decide on $b$. 
\end{proof}

\begin{lemma} \label{lem:co-ls}
	If the first round is of type $\ls$ or has a $\mult$ instruction with $\min$ as operation, then the algorithm does not solve agreement.
\end{lemma}

\begin{proof}
	Suppose that indeed in the first round we have a $\mult$ instruction with
	$\min$ operation or the first round is of type $\ls$. 
	We execute the phase under the  global predicate $\gp$.
	By Lemma~\ref{lem:co-bias-fire} we have $\set{a,b}\incl
	\fire_1(\bias(\th),\gp)$, for some sufficiently big $\th$.  
	Consider $\th_\fop=\max(\thru^\fop,\thr_\fop(\gp))+\e$
	for some small $\e$.
	Thanks to our proviso, the global predicate does not have a c-equalizer, 
	hence we can freely apply Lemmas~\ref{lem:co-any-frequency} 
	and~\ref{lem:co-any-bias} to get $\bias(\th_\fop)$ or
	$\biasq(\th_\fop)$ after round 	$\fo$.
	By Lemma~\ref{lem:co-no-mult-fo}, there is no $\mult$ instruction in round
	$\fop$.
	Hence $\set{b,?}\in\fire_\fop(\bias(\th_\fop),\gp)$.
	We can apply Lemma~\ref{lem:co-bias-propagation} to set $\dec$ of 
	one process to $b$ in this phase and leave the other processes undecided.
	Moreover, in the round $\fo$ the variable $\inp$ is set to $\bias(\max(\th,\th_{\fo+1}))$.
	
	In the next phase, 	Lemma~\ref{lem:co-bias-fire} says that  $\set{a,b}\in
	\fire_1(\bias(\max(\th,\th_{\fo+1})),\gp)$. 
	We can	get $\solo^a$ as the tuple after the first round under global predicate, hence we can set some $\dec$ to $a$.
\end{proof}

\begin{lemma}
	If the first round does not have a $\mult$ instruction then the algorithm does not terminate.
\end{lemma}

\begin{proof}
	Since the first round does not have type $\ls$, if there are no $\mult$ 
	instructions in the first round, then we have $\spread \lact{\f}_1 
	\solo_?$ for any predicate $\f$.
\end{proof}

\begin{lemma} \label{lem:co-ls-fo}
	If round $\fo+1$ is of type $\ls$, then the algorithm does not solve consensus.
\end{lemma}

\begin{proof}
	Suppose round $\fo+1$ is of type $\ls$. 
	We consider an execution of a phase under the
	global predicate and so we can freely use Lemmas~\ref{lem:co-any-frequency}
	and~\ref{lem:co-any-bias}. 
	We have seen in Lemma~\ref{lem:co-ls} that in the first round all the
	$\mult$ instructions must be $\smor$. We start with $\spread$. 
	We can then use Lemmas~\ref{lem:co-any-frequency} 
	and~\ref{lem:co-any-bias}
	to get  $\spread$ or $\spread^?$ after round $\fo$.
	In either case, because round $\fo+1$ is of type $\ls$
	and the global predicate does not have c-equalizers,
	it follows that we can get $\biasq(\th)$ for arbitrary $\th$
	after round $\fo+1$. Applying Lemma~\ref{lem:co-bias-propagation}
	we can make one process decide on $b$ and prevent the
	other processes from deciding.
	
	Notice that the state of $\inp$ in the next phase is still $\spread$.
	By Lemma~\ref{lem:co-spread} we have that $a \in \fire_1(\spread,\p)$.
	Hence we can get $\solo^a$ after the first round and use this 
	to make the undecided processes decide on $a$.
\end{proof}

\begin{lemma} \label{lem:co-constants}
	If the property of the constants is not satisfied, then the algorithm does not solve consensus.
\end{lemma}

\begin{proof}
	We consider an execution of a phase under the
	global predicate and so we can freely use Lemmas~\ref{lem:co-any-frequency}
	and~\ref{lem:co-any-bias}. 
	We have seen in Lemma~\ref{lem:co-ls} that in the first round all the
	$\mult$ instructions must be $\smor$. We start with $\spread$. 
	
	We have two cases, that resemble those of Lemma~\ref{lem:constants}.
	
	The first case is when all the rounds $1,\dots,\fo$ are not-c-preserving.
	Since we consider global predicate, there are not $c$-equalizers, so none of
	these rounds in an $\ls$ round. 
	This implies that all these rounds have a $\mult$ instruction. 
	We take $\th_\fop=\thru^{\fo+1}+\e$.
	From Lemma~\ref{lem:co-any-frequency} we can get $\bias(\th_\fop)$ as a tuple
	before the round $\fo+1$.
	By observation~\eqref{eq:syntactic-property}, we have
	$\thru^{\fo+1}\geq\thr_{\fo+1}(\gp)$, so $b\in\fire_{\fop}(\bias(\th_\fop),\gp)$. 
	We also have $?\in\fire_{\fop}(\bias(\th_\fop),\gp)$, because
	round $\fo+1$ does not have any $\mult$ instructions.
	Then we can apply Lemma~\ref{lem:co-bias-propagation} to set $\dec$ of 
	one process to $b$ in this round and leave the other processes undecided.

	The second case, is when there is a preserving round among $1,\dots,\fo$. 
	Let $j\leq \fo$ be the first such round.
	Since, all rounds before $j$ are not-c-preserving, by
	Lemma~\ref{lem:co-any-frequency}, we can get $\bias(\thru^j+\e)$ before
	round $j$. 
	Since $j$ is preserving, it is either of type $\ls$,  or has 
	no $\mult$ instructions or
	$\thr_j(\gp)<\max(\thru^j,\thr^{j,k}_m)$. 
	In the first case, since the global predicate does not have c-equalizers, we
	have $\set{b,?}\incl\fire_j(\bias(\thru^j+\e),\gp)$. 
	In the other cases the type of round $j$ can be $\every$ or $\lr$. 
	For type $\every$ we also get $\set{b,?}\incl\fire_j(\bias(\thru^j+\e),\gp)$. 
	For type $\lr$, we have that the round $j+1$ is $\ls$. 
	Since we have assumed that $\gp$ does not have a c-equalizer, $\gp_{j+1}$ does
	not have $\f_\ls$. 
	So we can get $\biasq(\th)$ for arbitrary $\th$ after round $j+1$.
	After all these cases we can use Lemma~\ref{lem:co-bias-propagation} to get
	$\biasq(\th_\fop)$ before round $\fop$; where as before $\th_\fop=\thru^\fop+\e$.
	As in the first case, we employ Lemma~\ref{lem:co-bias-propagation} to  make
	some process decide on $b$ and leave other processes undecided. 

	In both cases we can arrange the execution so that the state of
	$\inp$ after this phase is $(\bias(\th_\fop),\spreadq)$ or
	$(\spread,\spreadq)$. 
	The same argument as in Lemma~\ref{lem:constants} shows that some process can
	decide on $a$ in the next phase. 
\end{proof}

\begin{lemma} \label{lem:co-suff}
	If all the structural properties are satisfied then the algorithm satisfies agreement.
\end{lemma}

\begin{proof}
	It is clear that the algorithm satisfies agreement when the initial frequency is either $\solo$ or $\solo^a$.  Suppose $(\bias(\theta),d) \act{\p^*} (\bias(\theta'),d')$ such that for the first time in this transition sequence, some process has decided (say the process has decided on $a$). Since $\thrm^{1,k}/2 \geq 1-\thru^{\fo+1}$ we have that $\thru^{\fo+1} \ge 1/2$.
	Further, since round $\fo+1$ does not have any $\mult$ instructions (Lemma~\ref{lem:co-no-mult-fo}), it follows that every other process could only decide on $a$ or not decide at all. 
	Further notice that since $a$ was decided by some process and since
	round $\fo+1$ is not of type $\ls$ (Lemma~\ref{lem:co-ls-fo}), 
	it has to be the case that at least $\thru^{\fo+1}$ processes have $a$ as their $\inp$ value. Hence $\theta' < 1 - \thru^{\fo+1}$.
	
	Recall that the first round is not of type $\ls$ (Lemma~\ref{lem:co-ls}).
	Since $\theta' < 1 - \thru^{\fo+1} \le \thru^1$, it follows that $b$ cannot be fired from $\bias(\theta')$ using the $\uni$ instruction in the first round.
	Since $\theta' < 1 - \thru^{\fo+1} \le \thrmok/2$ and since every $\mult$ instruction in the first round has $\smor$ as its operator, it follows that
	$b$ cannot be fired from $\bias(\theta')$ using the $\mult$ instruction as well. 
	Hence the number of $b$'s can only decrease from this point onwards, and so it follows that no process from this point onwards can decide on $b$.
\end{proof}


\subsubsection*{Part 2: termination for coordinators}

\begin{lemma}\label{lem:co-fire-one}
	For the global predicate $\gp$: $a\in\fire_1(\bias(\th),\gp)$ iff $\th<\bthr$. 
	(Similarly $b\in\fire_1(\bias(\th),\gp)$ iff $1-\bthr<\th$).
\end{lemma}
\begin{proof}
	Since the first round cannot be a $\ls$ round (Lemma~\ref{lem:co-ls}), the
	proof of this lemma 
	is the same as that of Lemma~\ref{lem:fire-one}.
\end{proof}

\begin{corollary}\label{cor:co-no-a-above-bthr}
	For every predicate $\p$, if $\th\geq\bthr$ then $a\not\in\fire_1(\bias(\th),\p)$.
	Similarly if $\th\leq 1-\bthr$ then $b \not \in\fire_1(\bias(\th),\p)$.
\end{corollary}

\begin{lemma}\label{lem:co-unifier}
	Suppose $\p$ is a unifier and  $\bias(\th)\act{\p}f$.
	If $\thruo\leq \thrmok$ or $1-\bthr\le\th\le\bthr$ 
	then $f=\solo$ or $f=\solo^a$.
\end{lemma}
\begin{proof}
	The argument is the same as in Lemma~\ref{lem:unifier}, as the first round
	cannot be of type $\ls$. 
\end{proof}

\begin{lemma}\label{lem:co-decider}
	If $\p$ is a decider and $(\solo,\solo^?)\act{\p}(f',d')$ 
	then $(f',d') = (\solo,\solo)$.
	Similarly, if $(\solo^a,\solo^?) \act{\p}(f',d')$ then $(f',d') =
	(\solo^a,\solo^a)$.
	In case $\thrmok\leq \thruo$, for every $\th\geq\bthr$: if
	$(\bias(\th),\solo^?)\act{\p}(f',d')$ then $(f',d')=(\solo,\solo)$
	and for every $\th\leq 1-\bthr$: if
	$(\bias(\th),\solo^?)\act{\p}(f',d')$ then $(f',d')=(\solo^a,\solo^a)$.
\end{lemma}
\begin{proof}
	The same as in the case of the core language. 
\end{proof}

\begin{lemma}\label{lem:co-main-positive}
	If an algorithm in a core language has structural properties from
	Definition~\ref{def:structure}, and satisfies condition cT1 then it solves consensus.
\end{lemma}
\begin{proof}
	The same as for the core language.
\end{proof}

\subsubsection*{Part 3: non-termination for coordinators}

\begin{lemma}\label{lem:co-not-decider}
	If $\p$ is a not a c-decider then	
	$\solo\act{\p}\solo$ and $\solo^a \act{\p} \solo^a$; namely, no process may decide.
\end{lemma}
\begin{proof}
	If $\p$ is not a decider then there is a round, say $i$, that is not c-solo-safe. 
	By definition this means that either round $i$ is of type $\ls$ with $\p_i$
	not containing $\f_\ls$ or it has one of the two other types and $\thr_i(\p)< \thru^i$.
	It is then easy to verify that for $j < i$, $\solo \lact{\p_j}_j \solo$,
	$\solo \lact{\p_i}_i \solo^?$ and $\solo^? \lact{\p_k}_k \solo^?$ for $k > i$. Hence this
	ensures that no process decides during this phase.
	Similar proof holds when the $\inp$ tuple is $\solo^a$.
\end{proof}

\begin{lemma}\label{lem:co-global-th}
	For the global predicate $\gp$: if $1/2\leq \th\leq \bthr$ then
	$\bias(\th)\act{\gp}\bias(\th')$ for every $\th'\geq 1/2$.
\end{lemma}
\begin{proof}
	The proof follows the same argument as in Lemma~\ref{lem:global-th}. 
	There are some complications due to new types of rounds. 

	As in Lemma~\ref{lem:global-th} we start by observing that $\set{a,b}\incl
	\fire_1(\bias(\th,\gp)$. 
	This follows, as we have observed that the first round cannot be of type
	$\ls$. 

	If there are $\mult$ instructions in rounds $2,\dots,\fo$ then
	Lemma~\ref{lem:co-any-frequency} ensures that for arbitrary $\th'$ we can get
	$\bias(\th')$ after round $\fo$ (we have observed that round $\fo$ cannot be
	of type $\lr$). 
	Since there is no $\mult$ instruction in round $\fop$ (Proviso~\ref{proviso}
	and Lemma~\ref{lem:co-no-mult-fo}) we can get $\soloq$ after round $\fop$ and
	decide on nothing. 
	Hence we are done in this case.

	If some round $2,\dots,\fo$ does not have a $\mult$ instruction then we can
	use Lemma~\ref{lem:co-any-bias} to get $\biasq(\th'')$ as well as
	$\biasqa(\th'')$, for arbitrary $\th''$, after round $\fo$.
	There are two cases depending on $\th'\geq \th$ or not.

	If $\th'\geq \th$ then we take $\th''=\min(\th,\thru^\fop-\e)$. 
	For the same reasons as in Lemma~\ref{lem:global-th} that $\th''\geq 1/2$, so we can get
	$\bias(\th')$ as a state of $\inp$ after round $\fo$. 
	We show that $\bias(\th')\lact{\gp}_\fop\soloq$.
	If round $\fop$ is of type $\ls$ then this is direct from definition since
	round $\fo$ is not of type $\lr$. 
	Otherwise, we can just set the whole multiset of values to every process, and
	there is not enough of $b$'s to pass $\thru^\fop$ threshold.
	So we are done in this case.

	The remaining case is when $\th'<\th$. 
	As in Lemma~\ref{lem:global-th} we reach $\biasqa(\th'')$ for
	$\th''=\th-\th'$. 
	This gives us some $a$'s that we need, to convert $\bias(\th)$ to
	$\bias(\th')$. 
	As in the previous case we argue that we can get $\soloq$ after round $\fop$.
	So we are done in this case too. 
\end{proof}

\begin{lemma}
	If $\p$ is not a c-unifier then 
	\begin{equation*}
		\bias(\th)\act{\p}\bias(\th)\qquad \text{for some $1/2\leq \th <\bthr$}
	\end{equation*}
\end{lemma}
\begin{proof}
	As in the proof of an analogous lemma for the core language,
	Lemma~\ref{lem:not-uni}, we examine all the reasons for $\p$ not to be a
	c-unifier.
	
	The case of conditions of constants is the same as in Lemma~\ref{lem:not-uni}
	as the round cannot be of type $\ls$.
	If there is no equalizer in $\p$ up to round $\fo$ then the reasoning is the
	same but now using Lemmas~\ref{lem:co-any-frequency} and~\ref{lem:co-any-bias}.
\end{proof}

We can conclude the non-termination case. 
The proof is the same in for the core language but now using
Lemmas~\ref{lem:co-not-decider} and~\ref{lem:co-global-th}.

\begin{lemma}
	If the structural properties from Definition~\ref{def:structure} hold, but the
	condition cT1 does not hold then the algorithm does not terminate
\end{lemma}

\section{Proofs for algorithms with coordinators and timestamps}
We give a proof of Theorem~\ref{thm:ts-coordinators}. 
The structure of the proof is the same as in the other cases. 

\subsubsection*{Part 1: Structural properties for coordinators with timestamps}

\begin{lemma}
	If there is a round without uni instruction then the algorithm does not terminate.
\end{lemma}
\begin{proof}
	Let $l$ be the first round without uni instruction and let $\p$
	be any predicate. If $l < \fo$,
	we get $(\solo^a,0)\act{\p}(\solo^a,0)$ for every communication predicate.
	Otherwise we get $(\solo^a,i) \act{\p} (\solo^a,i+1)$ for every communication predicate.
\end{proof}

\begin{lemma} \label{lem:co-ts-weak-mult}
	If the first round is not of type $\ls$ then the first round should have a $\mult$ instruction.
\end{lemma}

\begin{proof}
	Let $\p$ be any predicate. If the first round is not of type $\ls$
	and does not have a $\mult$ instruction then we will have
	$(\spread,0) \act{\p} (\spread,0)$.
\end{proof}

The following lemma is an adaption of Lemma~\ref{lem:ts-min} to the extension
with coordinators. 

\begin{lemma}\label{lem:co-ts-no-mult-fo}
	If round $\fo$ is not a $\ls$ round 
	and either contains a $\mult$ instruction  or
	$\thru^{\fo} < 1/2$, then the algorithm 
	does not satisfy agreement, or we can remove the $\mult$ instruction and make $\thru^{\fo} = 1/2$ without altering the
	semantics of the algorithm.
\end{lemma}
\begin{proof}
	Suppose round $\fo$ is not of type $\ls$ and
	either contains a $\mult$ instruction or
	$\thru^{\fo} < 1/2$. We consider two cases:
	
	The first case is when there does not exist any tuple $(f,t)$ with
	$(f,t) \lact{\gp_1}_1 f_1 \dots \lact{\gp_{\fo-2}}_{\fo-2} f_{\fo-2}$
	such that $a,b \in \fire_{\fo-1}(f_{\fo-2},\gp)$. 
	Notice that this happens in particular when some round before round $\fo$ are
	of type $\ls$. 
	It is then clear that the $\mult$ instructions in round $\fo$ will never be
	fired and so we can remove all these instructions in round $\fo$.
	Further it is also clear that setting $\thru^{\fo} = 1/2$ 
	does not affect the semantics of the algorithm in this case.
	
	So it remains to examine the case when there exists a tuple 
	$(f,t)$ with
	$(f,t) \lact{\gp_1}_1  f_1\lact{\gp_2} \cdots\lact{\gp_{\fo-2}} f_{\fo-2}$ such that
	$a,b \in \fire_{\fo-1}(f_{\fo-2},\gp)$.
	By the above observation, none of the rounds before round $\fo$
	are of type $\ls$. 
	Further, we have assumed that round $\fo$ is itself not of type $\ls$.
	By our proviso, it also follows that round $\fo$ is not of type $\lr$.
	Since every $\lr$ round should be followed by a $\ls$ round,
	it follows that in this case all the rounds up to and including
	round $\fo$ are of type $\every$. Hence, the proof of this case is the
	same as the proof of Lemma~\ref{lem:ts-min}.
\end{proof}

The proof of the following Corollary is the similar to the proof
of corollary~\ref{cor:ts-min}.

\begin{corollary} \label{cor:co-ts-no-mult-fo}
	If round $\fo+1$ is not of type $\ls$ and has a $\mult$ instruction or $\thru^{\fo+1} < 1/2$,
	then the $\mult$ instruction can be removed and $\thru^{\fo+1}$ can be 
	made $1/2$ without altering the semantics of the algorithm.
\end{corollary}

\begin{lemma} \label{lem:co-ts-no-ls}
	The first round cannot be of type $\ls$.
\end{lemma}
\begin{proof}
	Suppose the first round is of type $\ls$. 
	We execute the phase under the  global predicate $\gp$.
	By semantics we have $\set{a,b}\incl
	\fire_1(\bias(\th),\gp)$, for arbitrary $\th$.  
	Consider $\th_\fop=\max(\thru^\fop,\thr_\fop(\gp))+\e$
	for some small $\e$.
	Thanks to our proviso, the global predicate does not have a c-equalizer, 
	hence we can freely apply Lemmas~\ref{lem:co-any-frequency} 
	and~\ref{lem:co-any-bias} to get $\bias(\th_\fop)$ or
	$\biasq(\th_\fop)$ after round 	$\fo$.
	By Corollary~\ref{cor:co-ts-no-mult-fo}, there is no $\mult$ instruction in round
	$\fop$.
	Hence $\set{b,?}\in\fire_\fop(\bias(\th_\fop),\gp)$.
	We can apply Lemma~\ref{lem:co-bias-propagation} to set $\dec$ of 
	one process to $b$ in this phase and leave the other processes undecided.
	Moreover, in the round $\fo$ the variable $\inp$ is set to $\bias(\max(\th,\th_{\fo+1}))$.
	
	In the next phase, once again we have $\set{a,b}\in
	\fire_1(\bias(\max(\th,\th_{\fo+1})),\gp)$. 
	We can	get $\solo_a$ under global predicate, hence we can set some $\dec$ to $a$.
\end{proof}

\begin{lemma}\label{lem:co-ts-mult-in-the-first-round}
	The first round should have a $\mult$ instruction.
\end{lemma}

\begin{proof}
	Follows from Lemmas~\ref{lem:co-ts-no-ls} and \ref{lem:co-ts-weak-mult}.
\end{proof}

\begin{lemma}\label{lem:co-ts-bias-fire}
	For the global predicate  $\p$, we have
	$\set{a,b}\incl\fire_1((\bias(\th),i),\p)$ for 
	sufficiently big $\th$ and every $i$. 
\end{lemma}

\begin{proof}
	Similar to that of Lemma~\ref{lem:ts-bias-fire}.
\end{proof}

\begin{lemma} \label{lem:co-ls-fo}
	If round $\fo+1$ is of type $\ls$, then the algorithm does not solve consensus.
\end{lemma}
\begin{proof}
	Suppose round $\fo+1$ is of type $\ls$. 
	We consider an execution of a phase under the
	global predicate and so we can freely use Lemmas~\ref{lem:co-any-frequency}
	and~\ref{lem:co-any-bias}. 
	We have seen that the first round cannot be of type $\ls$. 
	We can take $\th$ big enough to have $\set{a,b}\in\fire_1(\bias(\th),\gp)$. 
	We can then use Lemmas~\ref{lem:co-any-frequency} 
	and~\ref{lem:co-any-bias}
	to get  $\bias(\th')$ or $\biasq(\th')$ after round $\fo$; for arbitrary
	$\th'$. 
	In either case, because round $\fo+1$ is of type $\ls$
	and the global predicate does not have c-equalizers,
	it follows that we can get $\biasq(\th'')$ for arbitrary $\th''$
	after round $\fop$. Applying Lemma~\ref{lem:co-bias-propagation}
	we can make one process decide on $b$ and prevent the
	other processes from deciding.
	
	Notice that the state of $\inp$ in the next phase will have $\th'$ processes
	with value $b$ and timestamp $1$.
	Till now we have put no constraints on $\th'$, so we can take it sufficiently
	small so that $1-\th' > \thru^1$. This enables us to get $\solo^a$ after the first round.
	We use this  	to make the undecided processes decide on $a$.
\end{proof}

\begin{lemma} \label{lem:co-ts-constants}
	If the property of constants from Definition~\ref{def:ts-structure} 
	is not satisfied, then agreement is violated.
\end{lemma}

\begin{proof}
	The proof is similar to the one of Lemma~\ref{lem:ts-constants}.
	We consider an execution under the global predicate $\gp$, and employ
	Lemmas~\ref{lem:co-any-frequency} and~\ref{lem:co-ts-bias-fire}.
	We start from configuration $(\bias(\th_1),0)$ where
	$\th_1>\thruo$ big enough so that by Lemma~\ref{lem:co-ts-bias-fire} we get
	$\set{a,b} \incl\fire_1(\bias(\th_1),\gp)$.
	Observe that the first round
	cannot be of type $\ls$ by Lemma~\ref{lem:co-ts-no-ls} so we can get arbitrary
	bias after the first round. 
	
	Due to Lemma~\ref{lem:co-ts-no-mult-fo} we know that there
	is a c-preserving round before round $\fop$.
	Let $j\leq \fo$ be the first c-preserving round.
	Since all rounds before $j$ are non-c-preserving, by
	Lemma~\ref{lem:co-any-frequency} we can get $\bias(\thr^j_u+\e)$, as well as 
	$\bias(1-(\thr^j_u+\e))$ before round $j$
	(intuitively, we can get bias with many $b$'s or many $a$'s). 
	Since $j$ is preserving, it is either of type $\ls$ or $\thr_j(\gp)<
	\max(\thru^j,\thrm^{j,k})$. 
	If it is of type $\ls$ then we get $\set{a,b,?}\in
	\fire_j(\bias((\thr^j_u+\e)),\gp)$ since  $\gp\dar_j$ is not $c$-equalizer . 
	In the other cases we use $\bias(\thr^j_u+\e)$ if we want to get $\set{b,?}$
	and $\bias(1-(\thr^j_u+\e))$ if we want to get $\set{a,?}$.
	If $j$ is of type $\every$ we get it at round $j$.
	If $j$ is of type $\lr$ then we get it at round $j+1$ since round $j+1$ must be necessarily of
	type $\ls$. 
	We then use Lemma~\ref{lem:co-bias-propagation} to reach $\biasq(\th_\fop)$ or
	$\biasq(1-\th_\fop)$ before round $\fop$; where as before
	$\th_\fop=\thru^\fop+\e$.
	
	We have two cases depending on whether $\thrmok< 1-\thru^{\fo+1}$ or not.

	If $\thrmok< 1-\thru^{\fo+1}$ then we reach $\biasq(\th_\fop)$ before round
	$\fop$, and then make some processes decide on 	$b$.
	After this phase there are $1-\th_\fop$ processes with timestamp $0$.
	We can ensure that among them there is at least one with value $a$ and one
	with value $b$. 
	Since there is $\mult$ instruction in the first round (Lemma~\ref{lem:co-ts-mult-in-the-first-round}), in the next phase we
	send all the  values 	with timestamp $0$. 
	This way we get $\soloa$ after the first round, and make some process decide
	on $a$.

	The remaining case is when $\thrmok\geq 1-\thru^{\fo+1}$. 
	So we have $\thruo<1-\thru^{\fo+1}$, since we have assumed that that the
	property of constants from 	Definition~\ref{def:ts-structure} does not hold. 
	This time we choose to get $\biasqa(\th_\fop)$ after round $\fo$, and make some process
	decide on $a$. 
	Since we have started with $\bias(\th_1)$ we can arrange updates so that at
	the beginning of the next phase we have at least $\min(\th_1,1-\th_\fop)$ processes who
	have value $b$  with timestamp $0$. 
	But $\thruo<\min(\th_1,1-\th_\fop)$, so by sending $\ho$ set consisting of these
	$b$'s we reach $\solo$ after the first round and make some processes decide on
	$b$. 
\end{proof}

\begin{lemma}
	If all the structural properties are satisfied then the algorithm satisfies agreement.
\end{lemma}

\begin{proof}
	It is clear that the algorithm satisfies agreement when the initial frequency
	is either $\solo$ or $\solo^a$.  Suppose $(\bias(\theta),t,d) \act{\p^*}
	(\bias(\theta'),t',d')$ such that for the first time in this transition sequence,
	some process has decided (say the process has decided on $b$). 
	Since there exists an $\ls$ round in the algorithm, it follows 
	that every other process could only decide
	on $b$ or not decide at all. 
	Further notice that since $b$ was decided by some
	process, it has to be the case that more than $\thru^{\fo+1}$ processes have $b$ as
	their $\inp$ value and maximum timestamps.
	This means that the number of $a$'s in the configuration is  less than $1-\thru^{\fo+1}$.
	Also notice that since round $\fo$ has no $\mult$ instructions
	and $\thru^{\fo} \ge 1/2$, it follows that no process with 
	value $a$ has the latest timestamp.
	
	Since $\thruo \geq 1 - \thru^{\fo+1}$, it follows that $a$ cannot be fired
	from $\bias(\theta')$ using the $\uni$ instruction in the first round. Further
	since $\thrmok\geq  1 - \thru^{\fo+1}$ it follows that any HO set bigger
	than $\thrmok$ has to contain a value with the latest timestamp. 
	As no $a$ has the latest timestamp, $a$ cannot be fired from $\bias(\theta')$
	using the $\mult$ instruction as
	well. In consequence, the number of $b$'s can only increase from this point onwards and so it
	follows that no process from this point onwards can decide on $a$. A similar
	argument applies if the first value decided was $a$. 
\end{proof}

\subsubsection*{Part 2: termination for coordinators with timestamps.}

The proof for termination is very similar to the case of timestamps. 

\begin{lemma}\label{lem:co-ts-dec}
	If $\f$ is a c-decider and $(\solo,t,\solo^?)\act{\f}(f,t',d)$
	then $(f,d) = (\solo,\solo)$ for any ts-tuple $t$. Similarly if $(\solo^a,t,\solo^?) \act{\f}(f,t',d)$ then $(f,d) = (\solo^a,\solo^a)$.
\end{lemma}
\begin{proof}
	Immediate
\end{proof}

\begin{lemma}\label{lem:co-ts-str-uni}
	Suppose $\f$ is a strong c-unifier and $(\bias(\th),t)\act{\f}(f,t')$ then
	$f=\solo$ or $f=\solo^a$ (for every tuple of timestamps $t$).
\end{lemma}
\begin{proof}
	Let $i$ be the round with c-equalizer. 
	Till round $i$ we cannot produce $?$. 
	After round $i$ we have $\solo$ or $\soloa$.
	This stays till round $\fo$ as the rounds after $i$ are c-solo-safe.
\end{proof}

\begin{proof}\textbf{Main positive}
	
	Suppose there is a strong c-unifier followed by a c-decider. 
	After the strong c-unifier we have $\solo$ or $\soloa$ thanks to Lemma~\ref{lem:co-ts-str-uni}. 
	After c-decider all processes decide thanks to Lemma~\ref{lem:co-ts-dec}.
\end{proof}

\subsubsection*{Part 3: Non-termination for coordinators with timestamps}

\begin{lemma}\label{lem:co-ts-not-decider}
	If $\p$ is a not a c-decider then $(\solo,t)\act{\p}(\solo,t)$ and
	$(\solo^a,t)\act{\p}(\solo^a,t)$ for every tuple of timestamps $t$.
\end{lemma}
\begin{proof}
	If $\p$ is not a c-decider then there is a round that is not c-solo-safe. 
	So we can go to $\soloq$ both from $\solo$ and from $\soloa$.
	From $\soloq$ no process can decide.
\end{proof}

\begin{lemma} \label{lem:co-ts-not-str-uni}
	If $\p$ is not a strong c-unifier
	then $(\bias(\th),i) \act{\p} (\bias(\th),j)$ 
	is possible (for large enough $\th$, arbitrary $i$, and some $j$).
\end{lemma}

\begin{proof}
	Let $\th > \max(\thru^1,\thr_1(\p)) + \e$. Suppose $\p$ is not a strong c-unifier. 
	We do a case analysis. 
	
	Suppose $\thr_1(\p) < \thrm^{1,k}$ or $\thr_1(\p) < \thru^1$. 
	We can get $\solo^?$ after the first round and then 
	use this to not decide on anything and retain the input tuple.
	
	Suppose $\p$ does not have  an c-equalizer. In this case can we  
	apply Lemmas~\ref{lem:co-ts-bias-fire},~\ref{lem:co-ts-no-ls},~\ref{lem:co-any-bias} and~\ref{lem:co-ts-no-mult-fo} to conclude that 
	we can reach $\soloq$ before round $\fop$
	and so we are done, because nothing is changed after the phase. 
	
	The next possible situation is that $i$ is the first component of $\p$ that is a c-equalizer,
	and there is a c-preserving round, call it $j$ before round $i$ 	
	(it can be round $1$ as well).
	Every round before round $j$ is non-c-preserving, so it cannot be of type $\ls$.
	This is because non-c-preserving round of type $\ls$ is necessarily a c-equalizer,
	and the first c-equalizer is $i$. 
	So every round up-to $(j-1)$ has to be of type $\every$, and round $j$ can be of
	either of type $\every$ or of type $\lr$ (because $\lr$ round must be followed by $\ls$ round
	thanks to assumption on page~\pageref{assumption-ls-lr}). 
	In both cases, by Lemma~\ref{lem:co-any-frequency} we can get to $\bias(\th')$
	(where $\th' > \max(\thru^j,\thr_j(\p))$) before round $j$ 
	(Notice that if $j = 1$ then we
	need to reach $\bias(\th')$ with $\th' > \max(\thru^1,\thr_1(\p))$ which is
	where we start at). 
	
	If round $j$ is of type $\every$, then since it is preserving it is easy to
	see that $\bias(\th') 	\lact{\p_j}_j \soloq$.
	The remaining possibility is that round $j-1$ is of type $\lr$.
	We can get $\oneb$ after round $j-1$, and because round $j$ is
	necessarily of type $\ls$ and is not c-equalizer, we can get $\soloq$ after
	round $j$. 
	In both cases, as $j<\fo$ no process changes $\inp$ value, or decides in this phase.
	
	The remaining possibility for $\p$ not to be strong c-unifier is that there is
	$i$-th round that is a c-equalizer followed by a non-c-solo safe round $j\leq
	\fo$.
	It is clear that we can reach $\solo$ after round $i$ and using this get $\soloq$
	after round $j$. 
	Hence nothing will change in the phase giving a transition $(\bias(\th),i)
	\act{\p} (\bias(\th),i)$.
\end{proof}

\begin{proof}
	\textbf{Main non-termination}
	
	We show that if there is no strong c-unifier followed by a c-decider, then the
	algorithm will not terminate. We start with $(\bias(\th),0)$ where $\th$ 
	is large enough. If $\p$ is not a strong c-unifier then by 
	Lemma~\ref{lem:co-ts-not-str-uni}, for every $i$ transition $(\bias(\th),i)
	\act{\p} (\bias(\th),j)$ is possible for some $j$.
	Hence if there is no 	strong c-unifier in the communication predicate then the algorithm will not  terminate.
	
	Otherwise let $\p^l$ be the first strong c-unifier. 
	Notice that $\p^l$ is 
	not the global predicate. Till $\p^l$ we can maintain 
	$(\bias(\th),i)$ for some $i$.
	Suppose $\p^l$ is not a c-decider. 
	By Lemma~\ref{lem:co-ts-str-uni} the state
	after this phase will become $\solo$ or $\solo^a$. However since $\p^l$ is not a
	c-decider, we can choose to not decide on any value. Hence we get the transition
	$(\bias(\th),i) \act{\p} (\solo,i+1)$. 
	Now, since none of the 	Lemma~\ref{lem:co-ts-not-decider} we can have a
	transition where no decision 	happens.   
	Hence the algorithm does not terminate if there is no c-decider after a strong c-unifier.
\end{proof}

\section{Conclusions}

We have characterized all algorithms solving consensus in a fragment of the
Heard-Of model. 
We have aimed at a fragment that can express most important algorithms while trying to
avoid ad hoc restrictions (c.f. proviso on page~\pageref{proviso}).
The fragment covers algorithms considered in the context of
verification~\cite{MarSprBas:17,ChaMer:09} with a notable exception of
algorithms sending more than one variable.
In  this work we have considered only single phase algorithms while originally the
model permits also to have initial phases. 
We believe that this is not a severe restriction. 
More severe and technically important restriction is that we allow to use only
one variable at a time, in particular it is not possible to send pairs of
variables. 

One curious direction of further research would be to list all ``best'' consensus
algorithms under some external constraints; for example the constraints can come from some 
properties of an execution platform external to the Heard-Of model.
This problem assumes that there is some way to compare two algorithms. 
One guiding principle for such a measure could be efficient use of
knowledge~\cite{MosRaj:02,Mos:16}: at 
every step the algorithm does maximum it can do, given its knowledge of the state
of the system. 

This research is on the borderline between distributed computing and
verification. 
From a distributed computing side it considers quite a simple model, but gives a
characterization result. 
From a verification side, the systems are complicated because the number of processes
is unbounded, there are timestamps, and interactions are based on a fraction of
processes having a particular value. 
We do not advance on verification methods for such a setting. 
Instead, we observe that in the context considered here verification may be avoided.
We believe that a similar phenomenon can appear also for other problems than
consensus. 
It is also an intriguing question to explore how much we can enrich the current model and still get a
characterization. 
We conjecture that a characterization is possible for an extension with
randomness covering at least the Ben-Or algorithm.
Of course, formalization of proofs, either in Coq or Isabelle, for such
extensions would be very helpful.

\bibliography{papers}

\end{document}